%% file: scatteredUniversalityDLT.tex
\documentclass[envcountsame]{llncs}
\usepackage{amsmath,amssymb}
\usepackage[utf8]{inputenc}
\usepackage{todonotes}
\usepackage{tikz}

\usepackage{todonotes}

\def\ta{\mathtt{a}}
\def\tb{\mathtt{b}}
\def\tc{\mathtt{c}}
\def\td{\mathtt{d}}

\def\lsk{\left<}
\def\rsk{\right>}

\newcommand{\letters}{\text{alph}}

\def\nth#1{#1$^{\text{th}}$}

\def\N{\mathbb{N}}

\DeclareMathOperator{\ScatFact}{ScatFact}

\DeclareMathOperator{\id}{id}
\DeclareMathOperator{\ar}{ar}
\DeclareMathOperator{\Pref}{Pref}

\DeclareMathOperator{\last}{last}
\DeclareMathOperator{\nnext}{next}

\title{Scattered Factor-Universality of Words}

\author{Laura Barker\inst{1} \and Pamela Fleischmann\inst{1} \and Katharina Harwardt\inst{1} \and Florin Manea\thanks{Supported by the DFG grant MA 5725/2-1. 
F.M. thanks  Pawe\l{} Gawrychowski for his comments and suggestions.} \inst{2} \and Dirk Nowotka\inst{1}}

\authorrunning{L. Barker, P. Fleischmann, K. Harwardt, F. Manea, and D. Nowotka}

\institute{Kiel University, Germany, \email{stu97347@mail.uni-kiel.de, fpa@informatik.uni-kiel.de, stu120568@mail.uni-kiel.de, dn@informatik.uni-kiel.de} \and
University of G\"ottingen, Germany, \email{florin.manea@informatik.uni-goettingen.de}}

\begin{document}
\maketitle

\begin{abstract}
\input{abstract}

\end{abstract}

\section{Introduction}
\input{intro}

\section{Preliminaries}\label{prels}
\input{prelims}

\section{Testing Simon's Congruence}\label{Simon}
\input{simon}

\section{Scattered Factor Universality}\label{basic}

\input{universal}

\section{On Modifying the Universality Index}\label{queries}
\input{queries}

\section{Conclusion}
\input{conclusion}

\bibliographystyle{plain}
\bibliography{scatfact}

%
%

\end{document}

%% file: abstract.tex
A word $u=u_1\dots u_n$ is a scattered factor of a word $w$ if $u$ can be obtained from $w$ by deleting some of its letters: there exist the (potentially empty) words $v_0,v_1,..,v_n$ such that $w = 
v_0u_1v_1...u_nv_n$. The set of all scattered factors up to length $k$ of a word is called its full $k$-spectrum. Firstly, we show an algorithm deciding whether the $k$-spectra for given $k$  of two words are equal or not, running in optimal time. Secondly, we consider a notion of scattered-factors universality: the word $w$, with $\letters(w)=\Sigma$, is called 
$k$-universal if its $k$-spectrum includes all words 
of length $k$ over the alphabet $\Sigma$; we extend this notion to $k$-circular universality. After a series of preliminary combinatorial results, we present an algorithm computing, for a given $k'$-universal word $w$ the minimal $i$ such that $w^i$ is $k$-universal for some 
$k>k'$. Several other connected problems~are~also~considered.

%% file: intro.tex
A scattered factor (also called subsequence or subword) of a given word $w$ is a 
word $u$ such that there exist (possibly empty) words $v_0, \ldots, v_n, u_1, \ldots, 
u_n$ with $u = u_1 \ldots u_n$ and $w = v_0 u_1 v_1 u_2 \ldots u_n v_n$. 
Thus, scattered factors of a word $w$ are imperfect representations of $w$, obtained by removing some of its parts. As such, there is considerable interest in the relationship between a word and its 
scattered factors, both from a theoretical and practical point of view (cf. e.g., 
the chapter {\em Subwords} by J. Sakarovitch and I. Simon in 
\cite[Chapter 6]{Loth97} for an introduction to the combinatorial properties).
Indeed, in situations where one has to deal with input strings in which errors may occur, e.g.,  sequencing DNA or transmitting a digital 
signal, scattered factors form a natural model for the processed data as  parts of the input may be missing. This versatility of scattered factors is also highlighted by the many contexts in which this concept appears. For instance, in~\cite{Zetzsche16,HalfonSZ17,KuskeZ19}, various logic-theories were developed around the notion of scattered factors which are analysed mostly with automata theory tools and discussed in connection to applications in formal verification. On an even more fundamental perspective, there have been efforts to bridge the gap between the field of combinatorics on words, with its usual non-commutative tools, and traditional linear algebra, via, e.g., subword histories or Parikh matrices (cf. e.g.,~\cite{Mat04,Salomaa05,Seki12}) which are algebraic structures in which the number of specific scattered factors occurring in a word are stored. In an algorithmic framework, scattered factors are central in many classical problems, e.g., the longest common subsequence or the shortest common supersequence problems \cite{Maier:1978,BringmannK18}, the string-to-string correction problem \cite{Wagner:1974}, as well as in bioinformatics-related works \cite{ElzingaRW08}. 

In this paper we focus, for a given word, on the sets of scattered factors of a given length: the (full) $k$-spectrum of $w$ is the set containing all scattered factors of $w$ 
of length exactly $k$ (up to $k$ resp.). The total set of scattered 
factors (also called downward closure) of $w = \ta\tb\ta$ is $\{\varepsilon, \ta, \ta\ta, \ta\tb, \ta\tb\ta, \tb, \tb\ta \}$ and the 2-spectrum is $\{ 
\ta\ta, \ta\tb, \tb\ta \}$. The study of scattered factors of a fixed length of 
a word has its roots in \cite{Simon72}, where the relation $\sim_k$ (called 
Simon's congruence) defines the congruence of words that have the same full $k$-spectra.
Our main interest here lies in a special congruence class w.r.t. $\sim_k$: the class of 
words which have the largest possible $k$-spectrum. A word $w$ is called {\em $k$-universal} if its $k$-spectrum contains 
all the words of length $k$ over a given alphabet. That is, 
$k$-universal words are those words that are as rich as possible in terms of 
scattered factors of length $k$ (and, consequently, also scattered factors of 
length at most $k$): the restriction of their downward closure to words of 
length $k$ contains all possible words of the respective length, i.e., is a {\em 
universal} language. Thus $w=\ta\tb\ta$ is not $2$-universal since $\tb\tb$ is 
not a scattered factor of $w$, while $w'=\ta\tb\ta\tb$ is $2$-universal. 
Calling a words {\em universal} if its $k$-spectrum contains all possible words of length $k$, is rooted in formal language theory. The classical universality problem (cf. e.g., \cite{HolzerK11}) is whether a given language $L$ (over an alphabet $\Sigma$) is equal to $\Sigma^{\ast}$, where  $L$ can be given, e.g., as the language accepted by an automaton. A variant of this problem, called length universality, asks, for a natural number $\ell$ and a language $L$ (over $\Sigma$), whether $L$  contains all strings of length $\ell$ over~$\Sigma$.  See \cite{GawrychowskiRSS17} for a series of results on this problem and a discussion on its motivation, and \cite{Rampersad:2012,KrotzschMT17,GawrychowskiRSS17} and the references therein for more results on the universality problem for various types of automata. The universality problem was also considered for words \cite{martin1934,Bruijn46} and, more recently, for partial words \cite{ChenKMS17,GoecknerGHKKKS18} w.r.t. their factors. In this context, the question is to find, for a given $\ell$, a word $w$ over an alphabet $\Sigma$, such that each word of length $\ell$ over $\Sigma$ occurs exactly once as a contiguous factor of $w$. De Bruijn sequences~\cite{Bruijn46} fulfil this property, and have been shown to have many applications in various areas of computer science or combinatorics, see \cite{ChenKMS17,GoecknerGHKKKS18} and the references therein. As such, our study of scattered factor-universality is related to, and motivated by, this well developed~and~classical~line~of~research. 

While $\sim_k$ is a well studied congruence relation from language theoretic, 
combinatorial, or algorithmic points of view (see \cite{Simon72,Loth97,KufMFCS} and the references therein), the study of 
universality w.r.t. scattered factors seems to have been mainly carried out 
from a language theoretic point of view. In \cite{KarandikarKS15} as well as in 
\cite{CSLKarandikarS,journals/lmcs/KarandikarS19} the authors approach, in the context of studying the height of piecewise testable languages, the notion of $\ell$-rich words, 
which coincides with the $\ell$-universal words we define here; we will discuss 
the relation between these notions, as well as our preference to talk about 
universality rather than richness, later in the paper. A combinatorial study of scattered factors universality 
was started in \cite{dlt2019}, where a simple characterisation of $k$-universal 
binary words was given. In the combinatorics on words literature, more attention was given to 
the so called binomial complexity of words, i.e., a measure of the 
multiset of scattered factors that occur in a word, where each occurrence of 
such a factor is considered as an element of the respective multiset (see, 
e.g., \cite{RigoS15,FreydenbergerGK15,LeroyRS17a,Rigo19}).
As such, it seemed interesting to us to continue the work on scattered 
factor universality: try to understand better (in general, not only in 
the case of binary alphabets) their combinatorial properties, but, mainly, try 
to develop an algorithmic toolbox around the concept of ($k$-)universal words. 

\noindent{\bf Our results.} In the preliminaries we give the basic definitions and recall the arch factorisation introduced by Hebrard \cite{TCS::Hebrard1991}. Moreover we explain in detail the connection to richness introduced in \cite{KarandikarKS15}. 

In Section~\ref{Simon} we show one of our main results: testing whether two words have the same full 
$k$-spectrum, for given $k\in\N$, can be done in optimal linear time for words over ordered alphabets and improve and extend the results of~\cite{KufMFCS}. They also lead to an optimal solution over general alphabets.

In Section~\ref{basic} we prove that the arch factorisation can be computed in time linear w.r.t. the word-length and, thus, we can also determine whether a given word is $k$-universal. Afterwards, we provide several 
combinatorial results on $k$-universal words (over arbitrary alphabets); while  some of them follow in a rather straightforward way from the seminal work of 
Simon \cite{Simon72}, other require a more involved analysis. One such result is a characterisation of $k$-universal words by comparing the spectra of $w$ and $w^2$. We also investigate the similarities and differences of the universality if a word $w$ is repeated or $w^R$ and $\pi(w)$ resp. are appended to $w$, for a morphic permutation of the alphabet $\pi$. As consequences, we get a linear run-time algorithm for computing a minimal length scattered factor of $ww$ that is not a scattered factor of $w$. This approach works for arbitrary alphabets, while, e.g., the approach of \cite{TCS::Hebrard1991} only works for binary ones. We conclude the section by analysing the new notion of $k$-circular universality, connected to the universality of repetitions.

In Section \ref{queries} we consider the problem of modifying the 
universality of a word by repeated concatenations or deletions. Motivated by the fact that, 
in general, starting from an input word $w$, we could reach larger sets of 
scattered factors of fixed length by iterative concatenations of $w$, we show 
that, for a word $w$ a positive integer $k$, we can compute efficiently the minimal $\ell$ such that $w^\ell$ is $k$-universal. This result is extensible to sets of words. Finally, the shortest prefix or suffix we need to delete to lower the universality index of a word to a given number can be computed in linear time. Interestingly, in all of the algorithms where we are concerned with reaching $k$-universality we never effectively construct a $k$-universal word (which would take exponential time, when $k$ is given as input via its binary encoding, and would have been needed when solving these problems using, e.g., \cite{KufMFCS,ElzingaRW08}). Our algorithms run in polynomial time w.r.t. $|w|$, the length of the input word, and $\log_2 k$, the size of the representation of~$k$.

%% file: prelims.tex
Let $\N$ be the set of natural numbers and $\N_0=\N\cup\{0\}$.
Define $[n]$ as the set $\{1,\ldots, n\}$, $[n]_0 = [n]\cup\{0\}$
for an $n\in\N$, and $\N_{\geq n}=\N\backslash[n-1]$.
An alphabet $\Sigma$ is a nonempty finite set of symbols called {\em letters}. 
A {\em word} is a finite sequence of letters from $\Sigma$, thus an element of 
the free monoid $\Sigma^{\ast}$. Let $\Sigma^+=\Sigma^{\ast}\backslash\{\varepsilon\}$, where $\epsilon$ is the empty word. The {\em length} of 
a word $w\in\Sigma^{\ast}$ is denoted by $|w|$. For $k\in \mathbb{N}$ 
 define $\Sigma^{k} = \{ w \in \Sigma^* 
| |w| = k \}$ and $\Sigma^{\leq k}, \Sigma^{\geq k}$ analogously.
A word $u\in\Sigma^{\ast}$ is a {\em factor} of $w\in\Sigma^{\ast}$ if $w=xuy$ 
for some $x,y\in\Sigma^{\ast}$. If $x=\varepsilon$ (resp. $y=\epsilon$), $u$ is called a 
{\em prefix} (resp. {\em suffix} of $w$). Let $\Pref_k(w)$  be the prefix of $w$ of length $k\in\N_0$. The \nth{$i$} letter of $w\in\Sigma^{\ast}$ is denoted by 
$w[i]$ for $i\in[|w|]$ and set $w[i .. j] = 
w[i] w[i+1] \ldots w[j]$ for $1\leq i\leq j\leq |w|$.
Define the {\em reversal} of $w\in\Sigma^n$ by $w^R=w[n]\ldots w[1]$.
Set $|w|_{\ta}=|\{i\in[|w|]|\,w[i]=\ta\}|$ and $\letters(w)$ $= \{\ta \in \Sigma | |w|_\ta 
> 0 \}$ for $w\in\Sigma^{\ast}$.
For a word $u \in \Sigma ^*$ we define $u^0 = \varepsilon, u^{i+1} = u^i u$, 
for $i\in\N$. A word $w \in \Sigma ^*$ is called {\em power} (repetition)  of a 
word $u\in\Sigma^{\ast}$, if $w = 
u^t$ for some $t\in\N_{\geq 2}$.  
A word $u\in\Sigma^{\ast}$ is a {\em conjugate} of $w\in\Sigma^{\ast}$ 
if there exist $x,y\in\Sigma^{\ast}$ with $w=xy$ and $u=yx$. A function $\pi:\Sigma^{\ast}\rightarrow\Sigma^{\ast}$ is called {\em morphic permutation} if $\pi$ is bijective and $\pi(uv)=\pi(u)\pi(v)$ for all $u,v\in\Sigma^{\ast}$.

\begin{definition}
A word $v= v_1 \ldots v_k\in \Sigma^*$ is a  {\em scattered factor} of $w \in \Sigma^*$   if there exist $x_1, \ldots, 
x_{k+1}\in\Sigma^{\ast}$ such that $w = x_1 v_1 \ldots x_k v_k x_{k+1}$.
Let $\ScatFact(w)$ be the set of all scattered factors of $w$ and
define $\ScatFact_k(w)$ $(resp., \ScatFact_{\leq k}(w))$ as the set of all scattered factors of $w$ of 
length (resp., up to) $k\in\N$. A word $u\in\Sigma^{\ast}$ is a {\em common scattered factor} of 
$w,v\in\Sigma^{\ast}$, if $u\in\ScatFact(w)\cap\ScatFact(v)$; the word $u$ is an {\em 
uncommon scattered factor} of $w$ and $v$ (and {\em distinguishes} them) if $u$ is a scattered factor of 
exactly one of them.
\end{definition}

For $k\in\N_0$, the sets $\ScatFact_k(w)$ and $\ScatFact_{\leq k}(w)$ are also 
known as the $k$-spectrum and the full-$k$-spectrum of $w$ resp.. Simon 
\cite{Simon72} defined the congruence $\sim_k$ in which $u,v\in\Sigma^{\ast}$ are congruent if they
have the same full $k$-spectrum and thus the same $k$-spectrum. The {\em shortlex normal form} of a word $w\in\Sigma^{\ast}$ w.r.t. $\sim_k$, where $\Sigma$ is an ordered alphabet, is the shortest word $u$ with $u\sim_k w$ which is also lexicographically smallest (w.r.t. the given order on $\Sigma$) amongst all words $v\sim_k w$ with $|v|=|u|$. The maximal cardinality of a word's $k$-spectrum is $|\Sigma|^k$ and as shown in \cite{dlt2019} this is equivalent in the binary case to 
$w\in\{\ta\tb,\tb\ta\}^k$. The following  definition captures this property of a word in a generalised setting.

\begin{definition}
A word $w\in\Sigma^{\ast}$ is called {\em $k$-universal} (w.r.t. $\Sigma$), for $k\in\N_0$, if $\ScatFact_k(w)=\Sigma^k$. We abbreviate $1$-universal by {\em universal}. The {\em universality-index} $\iota(w)$ of $w\in\Sigma^{\ast}$ is the largest $k$ such that $w$ is $k$-universal.
\end{definition}

\begin{remark}
Notice that $k$-universality is always w.r.t. a given alphabet $\Sigma$: the word $\ta\tb\tc\tb\ta$ is $1$-universal for $\Sigma=\{\ta,\tb,\tc\}$ but it is not universal for $\Sigma\cup\{\td\}$. If it is clear from the context, we do not explicitly mention $\Sigma$. The universality of the factors of a word $w$ is considered w.r.t. $\letters(w)$.
\end{remark}

Karandikar and Schnoebelen introduced in \cite{CSLKarandikarS,journals/lmcs/KarandikarS19} the notion
of richness of words: $w\in\Sigma^{\ast}$ is {\em rich} (w.r.t. $\Sigma$) if $\letters(w)=\Sigma$ (and
poor otherwise) and $w$ is $\ell$-rich if $w$ is the concatenation of $\ell\in\N$ 
rich words. Immediately we get that a word is universal iff it is rich and moreover that a word is $\ell$-rich iff it is $\ell$-universal and a rich-factorisation, i.e., the factorisation of an $\ell$-rich word into $\ell$ rich words, can be efficiently obtained. However, we will use the name {\em $\ell$-universality} rather than {\em $\ell$-richness}, as richness defines as well, e.g. the property of a word $w\in\Sigma^n$ to have $n+1$ distinct palindromic factors, see, e.g., \cite{DroubayJP01,LucaGZ08}. 
As $w$ is $\ell$-universal iff $w$ is the concatenation of $\ell\in\N$ universal words it follows immediately that, if $w$ is over the ordered alphabet $\Sigma=\{1<2<\ldots< \sigma\}$ and it is $\ell$-universal then its shortlex normal form w.r.t. $\sim_\ell$ is $(1\cdot2 \cdots \sigma)^\ell$ (as this is the shortest and lexicographically smallest $\ell$-universal word).

The following observation leads to the next definition: the word $w=\ta\tb\tc\in\{\ta,\tb,\tc\}^{\ast}$ is $1$-universal and $w^s$ is $s$-universal for all $s\in\N$. But, $v^2=(\ta\tb\ta\tb\tc\tc)^2\in\{\ta,\tb,\tc\}^{\ast}$ is $3$-universal even though $v$ is only $1$-universal. Notice that the conjugate 
$\ta\tb\tc\tc\ta\tb$ of $v$ is $2$-universal.

\begin{definition}
A word $w\in\Sigma^{\ast}$ is called {\em $k$-circular universal} if a conjugate of $w$ is $k$-universal (abbreviate $1$-circular universal by circular universal). The {\em circular universality index} $\zeta(w)$ of $w$ is the largest $k$ such that $w$ is $k$-circular universal. 
\end{definition}

\begin{remark}\label{not-unique}
It is worth noting that, unlike the case of factor universality of words and partial words \cite{martin1934,Bruijn46,ChenKMS17,GoecknerGHKKKS18}, in the case of scattered factors it does not make sense to try to identify a $k$-universal word $w\in\Sigma^{\ast}$, for $k\in\N_0$, such that each word from $\Sigma^k$ occurs {\em exactly once} as scattered factor of $w$. Indeed for $|\Sigma|=\sigma$, if $|w|\geq k+\sigma$ then there exists a word from $\Sigma^k$ which occurs at least twice as a scattered factor of $w$. Moreover, the shortest word which is $k$-universal has length $k\sigma$ (we need $\ta^k\in\ScatFact_k(w)$ for all $\ta\in\Sigma$). As $k\sigma\geq k+\sigma$ for $k,\sigma\in \N_{\geq 2}$, all $k$-universal words have scattered factors occurring more than once: there exists $i,j\in [\sigma+1]$ such that $w[i]=w[j]$ and $i\neq j$. Then $w[i]w[\sigma+2..\sigma+k], w[j]
w[\sigma+2..\sigma+k]\in\ScatFact_k(w)$ and $w[i]w[\sigma+2..\sigma+k]=w[j]w[\sigma+2..\sigma+k]$.
\end{remark}

We now recall the arch factorisation, introduced by Hebrard in~\cite{TCS::Hebrard1991}.

\begin{definition}[\cite{TCS::Hebrard1991}]
For $w\in\Sigma^{\ast}$ the {\em arch factorisation} of $w$ is given by $w=\ar_w(1)\dots \ar_w(k)r(w)$ for a $k\in\N_0$ with $\ar_w(i)$ is universal and 
$\ar_w(i)[|\ar_w(i)|]\not\in\letters(\ar_w(i)[1\dots |\ar_w(i)|-1])$  for all $i\in[n]$, and 
$\letters(r(w))\subset\Sigma$. 
The words $\ar_w(i)$ are called {\em archs} of $w$, $r(w)$ is called the {\em rest}. Set $m(w)=\ar_w(1)[|\ar_w(1)|]$ $\dots \ar_w(k)[|\ar_w(k)|]$ as the word containing the unique last letters of each arch.
\end{definition}

\begin{remark}
If the arch factorisation contains $k\in\N_0$ archs, the word is $k$-universal, thus the equivalence of $k$-richness and $k$-universality becomes clear. Moreover if a factor $v$ of $w\in\Sigma^{\ast}$ is $k$-universal then $w$ is also $k$-universal: if $v$ has an arch factorisation with $k$ archs then $w$'s arch factorisation has at least $k$ archs (in which the archs of $v$ and $w$ are not necessarily related).
\end{remark}

Finally, our main results are of algorithmic nature. The computational model we use is the standard unit-cost RAM with logarithmic word size: for an input of size $n$, each memory word can hold $\log n$ bits. Arithmetic and bitwise operations with numbers in $[n]$ are, thus, assumed to take $O(1)$ time. Arithmetic operations on numbers larger than $n$, with $\ell$ bits, take $O(\ell/\log n)$ time. For simplicity, when evaluating the complexity of an algorithm we first count the number of steps we perform (e.g., each arithmetic operation is counted as $1$, no matter the size of the operands), and then give the actual time needed to implement these steps in our model. In our algorithmic problems, we assume that the processed words are sequences of integers (called letters or symbols, each fitting in $O(1)$ memory words). In other words, we assume that the alphabet of our input words is {\em an integer alphabet}. In general, after a linear time preprocessing, we can assume that the letters of an input word of length $n$ over an integer alphabet $\Sigma$ are in $\{1, \ldots , |\Sigma|\}$ where, clearly, $|\Sigma| \leq n$. For a more detailed discussion see, e.g.,~\cite{crochemore}.

%% file: simon.tex
Our first result extends and improves the results of Fleischer and Kufleitner \cite{KufMFCS}.

\begin{theorem}\label{simon_new}
(1). Given a word $w$ over an integer alphabet $\Sigma$, with $|w|=n$, and a number $k\leq n$, we can compute the shortlex normal form of $w$ w.r.t. $\sim_k$ in time $O(n)$. 
(2.) Given two words $w', w''$ over an integer alphabet $\Sigma$, with $|w'|\leq |w''|=n$, and a number $k\leq n$, we can test if $w'\sim_k w''$ in time $O(n)$.

\end{theorem}

\begin{proof}
The main idea of the algorithm is that checking $w'\sim_k w''$ is equivalent to checking whether the shortlex normal forms w.r.t. $\sim_k$ of $w'$ and $w''$ are equal.
To compute the shortlex normal form of a word $w\in\Sigma^n$ w.r.t. $\sim_k$ the following approach was used in \cite{KufMFCS} :
firstly, for each position of $w$ the $x$- and $y$-coordinates were defined. The $x$-coordinate of $i$, denoted $x_i$, is the length of the shortest sequence of indices $1\leq i_1<i_2<\ldots <i_t=i$ such that $i_1$ is the position where the letter $w[i_1]$ occurs  $w$ for the first time and, for $1<j\leq t$, $i_j$ is the first position where $w[i_j]$ occurs in $w[i_{j-1}+1..i]$. Obviously, if $\ta$ occurs for the first time on position $i$ in $w$, then $x_i=1$ (see \cite{KufMFCS} for more details). A crucial property of the $x$-coordinates is that if $w[\ell]=w[i]=\ta$ for some $i>\ell$ such that $w[j]\neq \ta$ for all $\ell+1\leq j\leq i-1$, then $x_i=\min \{x_\ell,x_{\ell+1},\ldots, x_{i-1}\} + 1.$ The $y$-coordinate of a position $i$, denoted $y_i$, is defined symmetrically: $y_i$ is the length of the shortest sequence of indices $n\geq i_1>i_2>\ldots >i_t=i$ such that $i_1$ is the position where the letter $w[i_1]$ occurs last time in $w$ and, for $1<j\leq t$, $i_j$ is the last position where $w[i_j]$ occurs in $w[i..i_{j-1}-1]$. Clearly, if $w[\ell]=w[i]=\ta$ for some $i<\ell$ such that $w[j]\neq \ta$ for all $\ell-1\geq j\geq i+1$, then $y_i=\min \{y_{i+1}, \ldots, y_{\ell-1},y_{\ell}\} + 1.$ 

Computing the coordinates is done in two phases: the $x$-coordinates are computed and stored (in an array $x$ with elements $x_1,\ldots,x_n$) from left to right in phase 1a, and the $y$-coordinates are stored in an array $y$ with elements $y_1,\ldots,y_n$ and computed from right to left in phase 1b (while dynamically deleting a position whenever the sum of its coordinates is greater then $k + 1$ (cf. \cite[Prop. 2]{KufMFCS})).
Then, to compute the shortlex normal form, in a third phase, labelled phase 2, if letters $\tb>\ta$ occur consecutively in this order, they are interchanged whenever they have the same $x$- and $y$-coordinates and the sum of these coordinates is $k + 1$ (until this situation does not occur anymore). 

We now show how these steps can be implemented in $O(n)$ time for input words over integer alphabets. For simplicity, let $x[i..j]$ denote the sequence of coordinates $x_{i},x_{i+1},\ldots,x_j$; $\min(x[i..j])$ denotes $\min\{x_i,\ldots,x_j\}$. It is clear that in $O(n)$ time we can compute all values $\last[i]=\max(\{0\}\cup \{j<i| w[j]=w[i]\})$.

{\em Firstly, phase 1a.} For simplicity, assume that $x_0=0$. While going with $i$ from $1$ to 
$n$, we maintain a list $L$ of positions $0=i_0< i_1<i_2<\ldots<i_t= i$ such that the following
property is invariant: $x_{i_{\ell-1}}<x_{i_\ell}$ for $1\leq \ell\leq t$ and $x_p \geq x_{i_\ell}$ for
all $i_{\ell-1}<p \leq i_\ell$. After each $i$ is read, if $\last[i]=0$ then set $x_i=1$; otherwise, determine 
$x_i=\min(x[\last[i]..i-1])+1$ by $L$, then append $i$ to $L$ and update $L$ accordingly so that its invariant property holds. This is done as follows: we go through the list $L$ from right to left (i.e., inspect the elements $i_{t}, i_{t-1}, \ldots$) until we reach a position $i_{j-1}<\last[i]$ or completely traverse the list (i.e., $i_{j-1}=0$). Let us note now that all elements $x_\ell$ with $i-1 \geq \ell \geq \last[i]$ fulfill $x_\ell\geq x_{i_j}$ and $i_j\geq \last[i]$. Consequently, $x_i=x_{i_j}+1$. Morover, $x_{i_{j+1}}\geq x_{i_j}+1$. As such, we update the list $L$ so that it becomes $i_1,\ldots,i_j,i$ (and $x_i$ is stored in the array~$x$). 

Note that each position of $w$ is inserted once in $L$ and once deleted (but never reinserted). Also, the time needed for the update of $L$ caused by the insertion of $i$ is proportional to the number of elements removed from the list in that step. Accordingly, the total time needed to process $L$, for all $i$, is $O(n)$. Clearly, this procedure computes the $x$-coordinates of all the positions of $w$ correctly. 

{\em Secondly, phase 1b.} We cannot proceed exactly like in the previous case, because we need to dynamically delete a position whenever the sum of its coordinates is greater than $k + 1$ (i.e., as soon as we finished computing its $y$-coordinate and see that it is $>k+1$; this position does not influence the rest of the computation). If we would proceed just as above (right to left this time), it might be the case that after computing some $y_i$ we need to delete position $i$, instead of storing it in our list and removing some of the elements of the list. As such, our argument showing that the time spent for inspecting and updating the list in the steps where the $y$-coordinates are computed amortises to $O(n)$ would not~work. 

So, we will use an enhanced approach. For simplicity, assume that $y_{n+1}=0$ and that every time we should eliminate position $i$ we actually set $y_i$ to $+\infty$. Also, let $y[i..j]$ denote the sequence of coordinates $y_{i},y_{i+1},\ldots,y_j$; note that some of these coordinates can be $+\infty$. Let $\min(y[i..j])$ denote the minimum in the sequence  $y[i..j]$. 
Similarly to what we did in phase 1a, while going with $i$ from $n$ to $1$, we maintain a list $L'$ of positions $n+1=i_0> i_1>i_2>\ldots>i_t\geq i$ such that the following property is invariant:
$y_{i_{\ell-1}} < y_{i_\ell}$ for $1\leq \ell\leq t$ and
$y_p \geq y_{i_\ell}$ for all $i_{\ell-1}>p \geq i_\ell$. 
In the current case, we also have that $y_p=+\infty$ for all ${i_t}> p\geq i$. 
The numbers $i_0, i_1, i_2, \ldots, i_t\geq i$ contained in the list $L'$ at some moment
in our computation define a partition of the universe $[1,n]$ in intervals: $\{1\}, \{2\},
\ldots,$ $\{i-1\}, [i,i_{t-1}-1], [i_{t-1},i_{t-2}-1], \ldots , [i_{1},i_{0}-1]$ for which
we define an {\em interval union-find} data structure \cite{gabow,union-find}; here the
singleton $\{a\}$ is seen as the interval $[a,a]$. According to \cite{union-find}, in our
model of computation, such a structure can be initialized in $O(n)$ time such that we can
perform a sequence of $O(n)$ \texttt{union} and \texttt{find} operations on it in $O(n)$
time, with the crucial restriction that one can only unite neighbouring intervals. We
assume that \texttt{find(j)} returns the bounds of the interval stored in our data
structure to which $j$ belongs. From the definition of the list $L'$, it is clear that,
before processing position~$i$ (and after finishing processing position $i+1$), 
$y_{i_\ell}=\min(y[i+1..i_{\ell -1}-1])$ holds. We maintain a new array 
$\nnext[\cdot]$ with 
$|\Sigma|$ elements: before processing position $i$, $\nnext[w[i]]$ is the smallest
position $j>i$ where $w[i]$ occurs after position~$i$, which was not eliminated 
(i.e., smallest $j>i$ with $y_j\neq +\infty$), or $0$ if there is no such position. 
Position $i$ is now processed as follows: let $[a,b]$ be the interval returned by $\mathtt{find}(\nnext[i])$. If $a=i+1$ then let $\min = y_{i_t}$; if $a>i+1$ then there exists $j$ such that $[a,b]=[i_j,i_{j-1}-1]$ and $t>j>0$, so let $\min=y_{j}$. Let now $y=\min + 1$, and note that we should set $y_i=y$, but only if $x_i+i\leq k+1$. So, we check whether $x_i+i\leq k+1$ and, if yes, let $y_i=y$ and set $\nnext[w[i]]=i$; otherwise, set $y_i = + \infty$ (note that position $i$ becomes, as such, irrelevant when the $y$-coordinate is computed for other positions). If $y_i = + \infty$ then make the union of the intervals $\{i\}$ and $[i+1,i_{t-1}-1]$ and start processing $i-1$; $L'$ remains unchanged. If $y_i \neq + \infty$ then make the union of the intervals $\{i\}, [i+1,i_{t-1}-1], \ldots, [i_{j+1},i_{j}-1]$ and start processing $i-1$; $L'$ becomes $i,i_j, i_{j-1},\ldots, i_0$. 

As each position of $w$ is inserted at most once in $L'$, and then deleted once (never reinserted), the number of list operations is $O(n)$. The time needed for the update of $L'$, caused by the insertion of $i$ in $L'$, is proportional to the number of elements removed from $L'$ in that step, so the total time needed (exclusively) to process $L$ is $O(n)$. On top of that, for each position $i$, we run one \texttt{find} operation and a number of \texttt{union} operations proportional to the number of elements removed from $L'$ in that step. Overall we do $O(n)$ \texttt{union} and \texttt{find} operations on the {\em union-find} data structure. This takes in total, for all $i$, $O(n)$ time (including the initialisation). Thus, the time complexity of phase 1b is~linear.

{\em Thirdly, phase 2.} Assume that $w_0$ is the input word {\em of this phase}. Clearly, $|w_0|=m\leq n$, and we have computed the coordinates for all its positions (and maybe eliminated some positions of the initial input word $w$). We partition in linear time $O(n)$ the interval $[1,m]$ into $2t+1$ (possibly empty) lists of positions $L_1,\ldots,L_{2t+1}$ such that the following conditions hold. Firstly, all elements of $L_i$ are smaller than those of $L_{i+1}$ for $1\leq i\leq 2t$. Secondly, for $i$ odd, the elements $j$ in $L_i$ have $x_j+y_j<k+1$; for each $i$ even, there exist $a_i,b_i$ such that $a_i+b_i=k+1$ and for all $j$ in $L_i$ we have $x_j=a_i, y_j=b_i$. Thirdly, we want $t$ to be minimal with these properties. We now produce, also in linear time, a new list $U$: for each $i\leq t$ and $j\in L_{2i}$ we add the triplet $(i,w[j],j)$ in $U$. We sort the list of triples $U$ (cf. \cite[Prop. 10]{KufMFCS}) with radix sort in linear time \cite{cormen}. After sorting it, $U$ can be decomposed in $t$ consecutive blocks $U_1$, $U_2, \ldots,$ $U_t$, where $U_i$ contains the positions of $L_{2i}$ sorted w.r.t. the order on $\Sigma$ (i.e., determined by the second component of the pair). As such, $U_i$ induces a new order on the positions of $w_0$ stored in $L_{2i}$. We can now construct a word $w_1$ by just writing in order the letters of $w_0$ corresponding to the positions stored in $L_i$, for $i$ from $1$ to $2t+1$, such that the letters of $L_i$ are written in the original order, for $i$ odd, and in the order induced by $U_i$, for $i$ even. Clearly, this is a correct implementation of phase $2$ which runs in linear time. The word $w_1$ is the shortlex normal form of~$w$.

Summing up, we have shown how to compute the shortlex normal form of a word in linear time (for integer alphabets). Both our claims follow.
\qed
\end{proof}

This improves the complexity of the algorithm reported in~\cite{KufMFCS}, where the problem was solved in $O(n|\Sigma|)$ time. As such, over integer alphabets, testing Simon's congruence for a given $k$ can be done in optimal time, that does not depend on the input alphabet or on $k$. When no restriction is made on the input alphabet, we can first sort it, replace the letters by their ranks, and, as such, reduce the problem to the case of integer alphabets. In that case, testing Simon's congruence takes $O(|\Sigma|\log |\Sigma| + n)$ time which is again optimal: for $k=1$, testing if $w_1 \sim_1 w_2$ is equivalent (after a linear time processing) to testing whether two subsets of $\Sigma$ are equal, and this requires $\Theta(|\Sigma| \log |\Sigma|)$ time~\cite{dobkin}.

%% file: universal.tex
In this section we present several algorithmic and combinatorial results.

\begin{remark}
Theorem \ref{simon_new} allows us to decide in linear time $O(n)$ whether a word $w$ over $\Sigma=\{1<2<\ldots <\sigma\}$ is $k$-universal, for a given $k\leq n, \sigma\in\N$. We compute the shortlex normal form of $w$ w.r.t. $\sim_k$ and check whether it is $(1\cdot 2 \cdots \sigma)^k.$
\end{remark}

We can actually compute $\iota(w)$ efficiently by computing its arch factorisation in linear time in $|w|$. Moreover this allows us to check whether $w$ is $k$-universal for some given $k$ by just checking if $\iota(w)\geq k$ or not.

\begin{proposition}\label{decomp}
Given a word $w\in\Sigma^n$, we can compute $\iota(w)$ in time $O(n)$.
\end{proposition}

\begin{proof}
We actually compute the number $\ell$ of archs 
in the arch factorisation. For a lighter notation, we use $u_i=\ar_w(i)$ for $i\in[\ell]_0$. The factors $u_i$ can be computed in linear time as follows. 
We maintain an array $C$ of $|\Sigma|$ elements, whose all elements are 
initially $0$, and a counter $h$, which is initially $|\Sigma|$. For simplicity, 
let $m_0=0$. We go through the letters $w[j]$ of $w[m_{i-1}+1..n]$, from left to 
right, and if $C[w[j]]$ equals $0$, we decrement $h$ by $1$ and set 
$C[w[j]]=1$. Intuitively, we keep track of which letters of $\Sigma$ we meet 
while traversing $w[m_{i-1}+1..n]$ using the array $C$, and we store in $h$ how 
many letters we still need to see. As soon as $h=0$ or $j=n$, we stop: set 
$m_i=j$ (the position of the last letter of $w$ we read), 
$u_i=w[m_{i-1}+1..m_i]$ (the $i^{th}$ arch), and $h=|\Sigma|$ again. If $j<n$ then reinitialise all 
elements of $C$ to $0$ and restart the procedure for~$i+1$. Note that if $j=n$ then $u_i$ is $r(w)$ as introduced in the definition of the arch factorization. 
The time complexity of computing $u_j$  is $O(|u_j|)$, because we process each 
symbol of $u_i=w[m_{i-1}+1..m_i]$ in $O(1)$ time, and, at the end of the 
procedure, we reinitialise $C$ in $O(|\Sigma|)$ time iff $u_i$ contained all 
letters of $\Sigma$, so $|u_i|\geq |\Sigma|$. The conclusion follows.  \qed
\end{proof}

The following combinatorial result characterise universality by repetitions.

\begin{theorem}\label{theorem_nonewsubs}
A word $w\in\Sigma^{\geq k}$ with $\letters(w)=\Sigma$ is $k$-universal for $k\in\N_0$ iff 
$\ScatFact_k(w^n)=\ScatFact_k(w^{n+1})$ for an $n\in\N$. Moreover we have $\iota(w^n)\geq kn$ if $\iota(w)=k$.
\end{theorem}

\input{theorem_nonewsubs.tex}

As witnessed by $w=\ta\ta\tb\tb\in\{\ta,\tb\}^{\ast}$, $\iota(w^n)$ can be greater than $n\cdot\iota(w)$: $w$ is universal, not $2$-universal but $w^2=\ta\ta\tb.\tb\ta.\ta\tb.\tb$ is $3$-universal. We study this phenomenon at the end of this section.
Theorem~\ref{theorem_nonewsubs} can also be used to compute an uncommon scattered factor of $w$ and $ww$ over arbitrary alphabets; note that the shortest such a factor has to have length $k+1$ if $\iota(w)=k$.

\begin{proposition}\label{minlength}
Given a word $w\in\Sigma^{\ast}$ we can compute in linear time $O(|w|)$ one of the uncommon scattered factors of $w$ und $ww$ of minimal length.
\end{proposition}

\input{minlength}

\begin{remark}
By Proposition~\ref{minlength}, computing the shortest uncommon scattered factor of $w$ and $ww$ takes optimal $O(n)$ time, which is more efficient than running an algorithm computing the shortest uncommon scattered factor of two arbitrary words (see, e.g., \cite{ElzingaRW08,KufMFCS}, and note that we are not aware of any linear-time algorithm performing this task for integer alphabets). In particular, we can use Theorem \ref{simon_new} to find by binary search the smallest $k$ for which two words have distinct $k$-spectra in $O(n \log n)$ time. In \cite{TCS::Hebrard1991} a linear time algorithm solving this problem is given for binary alphabets; an extension seems non-trivial.
\end{remark}

Continuing the idea of Theorem~\ref{theorem_nonewsubs}, we investigate even-length palindromes, i.e. appending $w^R$ to $w$. The first result is similar to Theorem~\ref{theorem_nonewsubs} for $n=1$. Notice that $\iota(w)=\iota(w^R)$ follows immediately with the arch factorisation.

\begin{corollary}\label{wwRuniv}
A word $w$ is $k$-universal iff $\ScatFact_k(w)= \ScatFact_k(ww^R)$.
\end{corollary}

In contrast to $\iota(w^2)$, $\iota(ww^R)$ is never greater than $2\iota(w)$.

\begin{proposition}\label{ww^R}
Let $w\in\Sigma^{\ast}$ be a palindrome and $u=\Pref_{\lfloor\frac{|w|}{2}\rfloor}(w)$ with $\iota(u)=k\in\N$. For $|w|$ even we have $\iota(w)=2k$ if $|w|$ even and for $|w|$ odd we get $\iota(w)=2k+1$ iff $w[\frac{n+1}{2}]\cup\letters(r(u))=\Sigma$.
\end{proposition}

\input{wwR.tex}

\begin{remark}
If we consider the universality of a word $w=w_1\dots w_m$ for $m\in\N$ with $w_i\in\{u,u^R\}$ for a given word $u\in\Sigma^{\ast}$, then a combination of the previous results can be applied. Each time either $u^2$ or $(u^R)^2$ occurs Theorem~\ref{theorem_nonewsubs} can be applied (and the results about circular universality that finish this section). Whenever $uu^R$ or $u^Ru$ occur in $w$, the results of Proposition~\ref{ww^R} are applicable.
\end{remark}

Another generalisation of Theorem~\ref{theorem_nonewsubs} is to investigate concatenations under permutations:  for a morphic permuation $\pi$ of $\Sigma$ can we compute $\iota(w\pi(w))$?

\begin{lemma}\label{piw}
Let $\pi:\Sigma^{\ast}\rightarrow\Sigma^{\ast}$ be a morphic permutation. Then $\iota(w)=\iota(\pi(w))$ for all $w\in\Sigma^{\ast}$ and especially the factors of the arch factorisation of $w$ are mapped by $\pi$ to the factors of the arch factorisation of $\pi(w)$.
\end{lemma}

\input{piw.tex}

By Lemma~\ref{piw} we have $2\iota(w)\leq\iota(w\pi(w))\leq2\iota(w)+1$. Consider the universal word $w=\ta\tb\tc\tb\ta$. For $\pi(\ta)=\tc$, $\pi(\tb)=\tb$, and $\pi(\tc)=\ta$ we obtain $w\pi(w)=\ta\tb\tc.\tb\ta\tc.\tb\ta\tb\tc.$ which is $3$-universal. However, for the identity $\id$ on $\Sigma$ we get that $w\id(w)$ is $2$-universal. We can show exactly the case when $\iota(w\pi(w))=2\iota(w)+1$.

\begin{proposition}\label{pimore}
Let $\pi:\Sigma^{\ast}\rightarrow\Sigma^{\ast}$ be a morphic permutation and $w\in\Sigma^{\ast}$ with the arch factorisation $w=\ar_w(1)\dots\ar_w(k)r(w)$ and $\pi(w)^R=\ar_{\pi(w)^R}(1)\dots$ $\ar_{\pi(w)^R}(k)r(\pi(w)^R)$ for an appropriate $k\in\N_0$. Then $\iota(w\pi(w))=2\iota(w)+1$ iff $\letters(r(w)r(\pi(w)^R))=\Sigma$, i.e. the both rests together are $1$-universal.
\end{proposition}

\input{pimore.tex}

Proposition~\ref{pimore} ensures that, for a given word with a non-empty rest, we can raise the universality-index of $w\pi(w)$ by one if $\pi$ is chosen accordingly.

\begin{remark}
Appending permutations of the word instead of its images under permutations of the alphabet, i.e. appending to $w$ abelian equivalent words, does not lead to immediate results as the universality depends heavily on the~permutation. If $w$ is  $k$-universal, a permutation $\pi$ may arrange the letters in lexicographical order, so $\pi(w)$ would only be $1$-universal. On the other hand, the universality can be increased by sorting the letters in $1$-universal factors: $\ta_1^m\ta_2^m\dots\ta_{|\Sigma|}^m$ for $\Sigma=\{\ta_1,\dots,\ta_{|\Sigma|}\}$ is $1$-universal but $(\ta_1\dots\ta_{|\Sigma|})^m$ is $m$-universal, for $m\in\N$.
\end{remark}

In the rest of this section we present results regarding circular universality. Recall that a word $w$ is $k$-circular universal if a conjugate of $w$ is $k$-universal. Consider $\Sigma=\{\ta,\tb,\tc,\td\}$ and $w=\ta\tb\tb\tc\tc\td\ta\tb\ta\tc\td\tb\td\tc$. Note that $w$ is not $3$-universal ($\td\td\ta\not\in\ScatFact_3(w)$) but $2$-universal. Moreover, the conjugate 
$\tb\tb\tc\tc\td\ta\tb\ta\tc\td\tb\td\tc\ta$ of $w$ is $3$-universal; 
accordingly, $w$ is $3$-circular universal.

\begin{lemma}\label{universalcircularuniversal}
Let $w\in\Sigma^{\ast}$. If $\iota(w)=k\in\N$ then $k\leq\zeta(w)\leq k+1$. Moreover if $\zeta(w)=k+1$ then $\iota(w)\geq k$.
\end{lemma}

\input{universalcircularuniversal.tex}

\begin{lemma}\label{middleterm}
Let $w\in\Sigma^+$. If $\iota(w)=k$ and $\zeta(w)=k+1$ then there exists $v,z,u\in\Sigma^{\ast}$ 
such that $w=vzu$, with $u,v\neq\varepsilon$ and $\iota(z)=k$.
\end{lemma}

\input{middleterm}

The following theorem connects the circular universality index of a word with the universality index of the repetitions of that word.

\begin{theorem}\label{circ2normal}
Let $w\in\Sigma^{\ast}$. If $\iota(w)=k$ and $\zeta(w)=k+1$ then $\iota(w^s)=sk+s-1$, for all $s\in\N$.
\end{theorem}

\input{circ2normal.tex}

The other direction of Theorem~\ref{circ2normal} does not hold for arbitrary 
alphabets: Consider the $2$-universal word $w=\tb\ta\tb\tc\tc\ta\ta\tb\tc$. We 
have that $w^2$
is $5$-universal but $w$ is not $3$-circular universal. Nevertheless, Lemma 
\ref{middleterm} helps us show that the converse of Theorem \ref{circ2normal} 
holds for binary alphabets:

\begin{theorem}\label{theoitercirc}
Let $w\in\{\ta,\tb\}^{\ast}$ with $\iota(w)=k$ and $s\in\N$. Then
$\iota(w^s)=sk+s-1$ if $\zeta(w)=k+1$ and $sk$ otherwise.
\end{theorem}

\input{theoitercirc.tex}

%% file: theorem_nonewsubs.tex
\begin{proof}
Assume firstly $w$ to be $k$-universal, i.e. we have $\ScatFact_k(w)=\Sigma^k$. This implies $\Sigma^k\subseteq\ScatFact_k(w^{n})$ for all $n\in\N$. 
On the other hand we have $\ScatFact_k(w^n)\subseteq\Sigma^k$ and thus $\ScatFact_k(w^n)=\Sigma^k=\ScatFact_k(w^{n+1})$ for all $n\in\N$ . 
For the second direction assume for a fixed $n\in\N$, $\ScatFact_k(w^n)=\ScatFact_k(w^{n+1})$. 
We prove firstly that $\ScatFact_k(w^n)=\Sigma^k$ holds. Let $v\in\Sigma^{\ast}$. If $v=\varepsilon$, we have $v\in\ScatFact_k(w^n)$. 
Let $|v|=\ell\in\N$ and assume that $\Sigma^{\ell-1}\subseteq\ScatFact_k(w^n)$. Thus $v[1\dots\ell-1]\in\ScatFact_k(w^n)$. 
By $\letters(w)=\Sigma$ we have $v\in\ScatFact_k(w^{n+1})=\ScatFact_k(w^n)$. If $n=1$ we have immediately $\ScatFact_k(w)=\Sigma^k$ and thus $w$ is $k$-universal. 
Consider $n\in\N_{\geq 2}$. Suppose now $\ScatFact_k(w^{n-1})\subset\Sigma^k$. Let $w^{n-1}=\ar_{w^{n-1}}(1)\dots\ar_{w^{n-1}}(\ell)r(w^{n-1})$ be the arch factorisation of $w^{n-1}$ for an appropriate $\ell\in[k-1]$. Choose $p\in[\ell]$ such that $w=\ar_w(1)\dots\ar_w(p)r'$ and $r'$ is a proper prefix of $\ar_{w^{n-1}}(p+1)$. Then $r'$ is a suffix of $w^{n-1}$ and $w^n$. Especially $r(w)$ is a suffix of $r'$. Choose $\ta\in\Sigma\backslash\letters(r')$, i.e. $\ta\not\in\letters(r(w))$. By \cite[Propostion~2]{TCS::Hebrard1991}, $m[w^{n-1}]\ta\not\in\ScatFact_k(w^n)$. By $\letters(w)=\Sigma$ we have on the other hand $m[w^{n-1}]\ta\in\ScatFact_k(w^{n+1})$ - a contradiction. Thus we have $\ScatFact_k(w^{n-1})=\Sigma^k$. Inductively we get $\ScatFact_k(w)=\Sigma^k$ and thus $w$ is $k$-universal.

\medskip

For the second claim, we get immediately that $w^n$ is at least $kn$-universal if $\iota(w)=k$,  since the arch factorisation of $w$ occurs in each $w$ of the repetition.\qed
\end{proof}

%% file: minlength.tex
\begin{proof}
Just like in the proof of Theorem \ref{decomp}, we compute the decomposition (arch factorisation) $w=u_1\ldots u_k$ such that, for $i\in [k-1]$, the factor $w[1..m_i]=u_1\cdots u_i$ is the shortest prefix of $w$ such that $\Sigma^i \subseteq \ScatFact_i (w[1..m_i])$, and $u_k$ (called in the arch factorisation {\em the rest}) either does not contain all letters of $\Sigma$ or it does, but if we remove its last letter then it does not contain anymore all letters of $\Sigma$, i.e., $\Sigma^k \subseteq \ScatFact_k (w)$ but $\Sigma^k \not \subseteq \ScatFact_k (w[1..n-1])$. 

If $u_k$ does not contain all letters of $\Sigma$, then $k>1$ (as $w$ contains all letters of $\Sigma$). The procedure described in the proof of Theorem \ref{decomp} identifies a letter $\ta$ that does not occur in $u_k$. We construct the word $x=w[m_1]w[m_2]\cdots w[m_{k-1}]\ta = m(w)\ta$ (where $m(w)$ is defined w.r.t. the arch factorisation). Then $x$ is not a scattered factor of $w$ (and all shorter words are scattered factors of $w$), but $x$ is scattered factor of $ww$ (as $\ta$ occurs in $w$, because $k>1$). Indeed, if $x$ were a scattered factor of $w$, then its \nth{$i$} letter should correspond to the letter occurring on position $j_i\geq m_i$ of $w$. 
This is clear for $m_1$: if $w[m_1]$  occurred also to the left of $m_1$ in $w$, then $u_1$ would not be the shortest prefix of $w[1..n]$ that contains all letters of $\Sigma$. 
Then, for $i\geq 1$, assume the property holds for the first $i$ letters of $x$. We show it for $i+1$. So, $x[i+1]$ should correspond to a letter of $w$ occurring to the right of $x[i]$. 
So on a position strictly greater than $m_i$. But $x[i+1]=w[m_{i+1}]$ occurs of the first time to the left of $m_i$ on position $m_{i+1}$. So, our statement is correct. Now, if the \nth{$(k-1)$} letter of $x$ occurs 
on a position greater or equal to $m_{k-1}$, then the last letter of $x$, namely $\ta$, should occur in $u_k=w[m_{k-1}+1..m_k]$, a contradiction.

If $u_k$ contains all letters of $\Sigma$, then let $x=w[m_1]w[m_2]\cdots w[m_{k}]\ta=m(w)\ta$, for some $\ta\in \Sigma$. Just like before, we can show that $x$ is not a scattered factor of $w$, but all shorter words are scattered factors of $w$; also $x$ is clearly a scattered factor of $ww$. 

Running the procedure described in Theorem \ref{decomp} takes linear time, and constructing $x$ also takes linear time. The conclusion follows.\qed
\end{proof}

%% file: wwR.tex
\begin{proof}
Consider firstly $|w|\equiv_20$, i.e. $w=uu^R$. By $\iota(u)=k$, $u$ has an arch factorisation with $k$ factors which also occur in $u^R$. This implies $\iota(ww^R)\geq 2k$. Suppose $\iota(uu^R)=2k+1$. Let $uu^R=\ar_{uu^R}(1)\dots \ar_{uu^R}(2k+1)r(uu^R)$ be the arch factorisation.  Since 
$k$ is maximal, $\ar_{uu^R}(1)\dots \ar_{uu^R}(k+1)$ is not a prefix of $u$, i.e. $\ar_{uu^R}(k+2)$ is a factor of $u^R$ and thus $\ar_{uu^R}(k+2)\dots \ar_{uu^R}(2k+1)r(uu^R)$ is a suffix of 
$u^R$. Hence we get $\ar_{uu^R}(k+1)=r(u)y$ for a prefix $y$ of $u^R$. If $|r(u)|=|y|$ we have $r(u)=y^R$ and thus $\Sigma=\letters(\ar_{uu^R}(k+1))=\letters(r(u))\subset\Sigma$.
If $|r(u)|<|y|$, then $r(u)^R$ is a prefix of $y$. This implies $\Sigma=\letters(\ar_{uu^R}(k+1))=\letters(y)$ and consequently we found an arch factorisation of $w$ (the second one) with $k+1$ factors. Consider now $|r(u)|>|y|$. Then $y^R$ is a suffix of $r(u)$ but by the definition of the arch factorisation $y[|y|]$ does not occur in $r(u)[1\dots |r(u)|-1]$. Since we get a contradiction in all three cases, the claim is proven for even-length palindromes.

\medskip

By a similar argument odd-length palindromes have to have exactly the letter in the middle which is missing in $r(u)$ to be $1$-universal.\qed
\end{proof}

%% file: piw.tex
\begin{proof}
Let $w\in\Sigma^{\ast}$ and $w=\ar_w(1)\dots\ar_w(k)r(w)$ be the arch factorisation of $w$ for an appropriate $k\in\N_0$. 
By the definition of the arch factorisation $\ar_w(i)[|\ar_w(i)|]$ does not occur in $\ar_w(i)[1\dots |\ar_w(i)|-1]$ for all $i\in[k]$.
Set $k_i=\sum_{j=1}^i|\ar_w(j)|$ for $i\in[k]$.
Thus $\pi(\ar_w(i)[|\ar_w(i)|])$ occurs only once in $\pi(w)[k_i+1\dots k_{i+1}]$ and exactly as the last letter. Hence $\pi(\ar_w(1))\dots\pi(\ar_w(k))\pi(r(w))$ is the arch factorisation of $\pi(w)$. The other direction follows by applying $\pi^{-1}$ as a permutation to $\pi(w)$.\qed
\end{proof}

%% file: pimore.tex
\begin{proof}
Consider firstly that $r(w)r(\pi(w)^R)$ is $1$-universal. Then we get 
\[
w\pi(w)=\ar_w(1)\dots\ar_w(k)r(w).r(\pi(w)^R)^R(\ar_{\pi(w)^R}(k))^R\dots (\ar_{\pi(w)^R}(1))^R.
\]
Since all archs are $1$-universal by definition, the assumption implies that $w\pi(w)$ is $(2\iota(w)+1)$-universal and thus $\iota(w\pi(w))\geq 2\iota(w)+1$. The equality follows by the definition of $\iota$. For the other direction assume $\iota(w\pi(w))=2\iota(w)+1$. Here, we get the arch factorisation
\[
w\pi(w)=\ar_{w\pi(w)}(1)\dots\ar_{w\pi(w)}(2\iota(w)+1)r(w\pi(w)).
\]
This implies
\begin{multline*}
\ar_{w\pi(w)}(1)\dots\ar_{w\pi(w)}(2\iota(w)+1)r(w\pi(w))\\
=\ar_w(1)\dots \ar_{w}(k)r(w)r(\pi(w)^R)^R(\ar_{\pi(w)^R}(k))^R\dots (\ar_{\pi(w)^R}(1))^R.
\end{multline*}
By  $\iota(w)=k$ only the first $k$ archs can be contained in $w$. 
This implies that $r(w)$ is a prefix of $\ar_{w\pi(w)}(k+1)$. Choose $y\in\Sigma^+$ with $\ar_{w\pi(w)}(k+1)=r(w)y$.
By Lemma~\ref{piw} we have $\ar_w(i)=\pi(\ar_{\pi(w)}(i))$ and thus $y=r(\pi(w)^R))^R$. By $\ar_{w\pi(w)}(k+1)=\Sigma$ the claim is proven.\qed
\end{proof}

%% file: universalcircularuniversal.tex
\begin{proof}
Since $w$ is a conjugate of itself, $w$ is at least $k$-circular universal. 
Suppose $\zeta(w)=k+2$. Choose 
$x,y\in\Sigma^{\ast}$ with $w=xy$ and $yx=\ar_{yx}(1)\dots\ar_{yx}(k+2)r(yx)$. Since $\iota(w)=k$ there is no $i$ such that $y=w_1\cdots w_i$ (otherwise $w=xy$ 
would be $(k+1)$-universal).  Thus there exists a $j\in[k+2]$ and a 
proper prefix $y_1$ of $w_j$ such that $y=w_1\cdots w_{j-1}y_1$; let $x_1$ be such that $w_j=y_1x_1$. This implies  $w=xy=x_1w_{j+1}\dots w_{k+2}w_1\dots 
w_jy_1$ and we get that $k+1$ archs are contained in 
$w$. This contradicts the maximality of $k$.

For the second claim let $w=xy$ and $yx=\ar_{yx}(1)\dots\ar_{yx}(k+1)r(yx)$. If $y$ contains all archs then $\iota(w)=k+1$. If $y$ does not contain all archs, there exists an $i\in[k+1]$ such that a prefix of $\ar_{yx}(i)$ is a suffix of $y$ and the corresponding suffix of $\ar_{yx}(i)$ is a prefix of $x$. Thus $\ar_{yx}(1)\dots\ar_{yx}(i-1)\ar_{yx}(i+1)\ar_{yx}(k+1)$ is a scattered factor of $w$.\qed
\end{proof}

%% file: middleterm.tex
\begin{proof}
By $\zeta(w)=k+1$ there exist 
$x,y\in\Sigma^{\ast}$ with $w=xy$ and $yx=\ar_{yx}(1)\dots$ 
$\ar_{yx}(k+1)r(yx)$. Since $\iota(w)=k$ there is no $i$ such that $y=w_1\cdots w_i$ (otherwise $w=xy$ 
would be $(k+1)$-universal). 
Thus, there exists 
$i\in[k+1]_0$ with $w_{i+1}=uv$ and $u$ is a proper and non-empty suffix of 
$y$ 
and $v$ is a proper and non-empty prefix of $x$ with 
$\letters(u),\letters(v)\subset\Sigma$. This implies 
\[
y
w=xy=v\ar_w(i+2)\dots \ar_w(k+1)\ar_w(1)\dots \ar_w(i)u.
\]
Let $z=\ar_w(i+2)\dots \ar_w(k+1)\ar_w(1)\dots \ar_w(i)$. Clearly, $z$ contains  
$1$-universal words, so $\iota(z)\geq k$. By $\iota(w)=k$ follows immediately $\iota(z)\leq k$.\qed
\end{proof}

%% file: circ2normal.tex
\begin{proof}
By Lemma~\ref{middleterm} there exist $v,z,u\in\Sigma^{\ast}$ with $w=vzu$, $\iota(z)=k$, and $\letters(v)$, $\letters(u)\subset\Sigma$. Consequently we have that
\[
w^s=(vzu)^s=v(zuv)^{s-1}zu
\]
is $((s-1)(k+1)+k)$-universal, thus $\iota(w^s)\geq(sk+s-1)$.
Since $\iota(w)=k$ and $w^s$ only contains $s-1$ transitions from one $w$ to another, $w^s$ cannot have a higher universality.\qed
\end{proof}

%% file: theoitercirc.tex
\begin{proof}
By Theorem~\ref{circ2normal} it suffices to prove $\zeta(w)=k+1$ if $w^s$ is $(sk+1)$-universal. Assume $\iota(w)\geq sk+1$. If for all conjugates $v$ of $w$ we have $v[1]\neq v[|w|]$ then $w$ is of even length and we have $w=( \ta\tb)^{k}$ or $w=(\tb\ta)^{k}$; this implies immeditaly $\zeta(w)=k$ by the arch factorisation.
Thus we know that there exists a conjugate $v$ of $w$ with $v[1]=v[|w|]$. Since $w^s$ is a conjugate of $v^s$ and $w^s$ is $(sk+1)$-universal, $v^s$ is $(sk+1)$-circular universal. By Lemma~\ref{universalcircularuniversal} we get that $v^s$ is $(sk)$-universal and by Theorem~\ref{circ2normal} follows that $v^{2s}$ 
is $(2sk+1)$-universal. By \cite[Theorem4]{dlt2019} $v^{2s}$ contains $2sk+1$ disjoint occurrences of $\ta\tb$ or $\tb\ta$.
By $v[1]=v[n]$ non of these occurrences can start in one $v$ and end in the following. This implies that one $v$ contains $k+1$ of these occurrences and therefore $\iota(v)\geq k+1$. Hence we get $\zeta(w)=k+1$.\qed
\end{proof}

%% file: queries.tex
In this section we present algorithms answering the for us most natural questions regarding universality: is a specific factor $v$ of $w\in\Sigma^{\ast}$ universal? what is the minimal $\ell\in\N$ such that $w^{\ell}$ is $k$-universal for a given $k\in\N$? how many (and which) words from a given set do we have to concatenate such that the resulting word is $k$-universal for a given $k\in\N$? what is the longest (shortest) prefix (suffix) of a word being $k$-universal for a given $k\in\N$? In the following lemma we establish  some preliminary data structures.

\begin{lemma}\label{help1}
Given a word $x\in\Sigma^n$ with $\letters(x)=\Sigma$, we can compute in $O(n)$ and for all $j\in[n]$
\begin{itemize}
\item the shortest $1$-universal prefix of $x[j..n]$: $u_x[j]=\min\{i\mid x[j..i]$ is universal$\}$,
\item the value $\iota(x[j..n])$: $t_x[j]=\max\{t\mid \ScatFact_t(x[j..n])= \Sigma^t\}$, and 
\item the minimal $\ell\in[n]$ with $\iota(x[j..\ell])=\iota(x[j..|x|])$: $m_x[j]=\min\{i\mid \ScatFact_{t_x[j]}$ $(x[j..i])= \Sigma^{t_x[j]}\}$.
\end{itemize}
\end{lemma}

\input{help1.tex}

The data structures constructed in Lemma \ref{help1} allow us to test in $O(1)$ time~the universality of factors $w[i..j]$ of a given word $w$, w.r.t. $\letters(w)=\Sigma$: $w[i..j]$ is $\Sigma$-universal iff $j\geq u_w[i]$.
The combinatorial results of Section~\ref{basic} give us an initial idea on how 
the universality of repetitions of a word relates to the universality of that 
word: Theorem~\ref{circ2normal} shows that in order to compute 
the minimum $s$ such that $w^s$ is $\ell$-universal, for a given {\em binary} 
word $w$ and a number $\ell$, can be reduced to computing the circular 
universality of~$w$. Unfortunately, this is not the case for all alphabets, as also shown in Section~\ref{basic}. However, this number $s$ can be computed efficiently, for input words over alphabets of all 
sizes. While the main idea for binary alphabets was to analyse the universality 
index of the conjugates of $w$ (i.e., factors of length $|w|$ of $ww$), in the general case we can analyse the universality index of the suffixes of $ww$, by 
constructing the data structures of Lemma~\ref{help1} for $x=ww$. The problem is 
then reduced to solving an equation over integers in order to 
identify the smallest $\ell$ such that $w^\ell$ is $k$-universal. 

\begin{proposition}\label{theomin}
Given a word $w\in\Sigma^n$ with $\letters(w)=\Sigma$ and $k\in\N$, we can 
compute  the minimal $\ell$ such that $w^\ell$ is $k$-universal in $O(n+ {\frac{\log k}{\log n}})$ time.
\end{proposition}

\input{theomin.tex}

We can extend the previous result to the more general (but less motivated) case of arbitrary concatenations of words from a given set, not just repetitions of the same~word. The following preliminary results can be obtained. In all cases we give the number of steps of the algorithms, including arithmetic operations on $\log k$-bit numbers; the time complexities of these algorithms is obtained by multiplying these numbers by $O(\frac{\log k}{\log n})$.

%
%

\bigskip

\input{appendixlastpart.tex}

Finally, we consider the case of decreasing the universality of a word by an operation opposed to concatenation, namely the deletion of a prefix or a suffix.

\begin{theorem}\label{del-pref-suf}
Given $w\in\Sigma^n$ with $\iota(w)=m$  and a number 
$\ell<m$, we can compute in linear time the shortest prefix (resp., suffix) 
$w[1..i]$ (resp., $w[i..n]$) such that $w[i+1..n]$ (resp., $w[1..i-1]$) has 
universality index $\ell$. 
\end{theorem}

\input{del-pref-suf.tex}

Theorem \ref{del-pref-suf} allows us to compute which is the shortest prefix (suffix) we should delete so that we get a string of universality index $\ell$. Its proof  is based on the data structures of Lemma \ref{help1}. For instance, to compute the longest prefix $w[1..i-1]$ of $w$ which has universality index $\ell$, we identify the first $\ell+1$ factors of the decomposition of Theorem \ref{decomp}, assume that their concatenation is $w[1..i]$, and remove the last symbol of this string. A similar approach works for suffixes.

%% file: help1.tex
\begin{proof}
For each $j\in [n]$ and letter $\ta\in \Sigma$, denote $g_{\ta}[j]=\min\{i\mid i\geq j, w[i]=\ta\}$ (by convention, $g_\ta[j]=+\infty$ if $\ta$ does not occur in $x[j..n]$). Clearly, $u_x[j]=\max\{g_\ta[j]\mid \ta\in \Sigma\}$ holds, i.e., $u_x[j]$ is the end position of the shortest word starting on position $j$ in $x$ which contains all letters of~$\Sigma$ (the value $g_\ta[j]$ is strongly related to the value $X_\ta(w[j..n])$ - read "next $\ta$ in $w[j..n]$"- used in \cite{KufMFCS} to denote the first position of $\ta$ in $w[j..n]$).
It is essential to note that we will not compute all the values $g_{\ta}[j]$, but only the values $u_x[j]$, for all $j$. As such, $x[j..u_x[j]]$ is the shortest universal prefix of $x[j..n]$.

Computing the elements of $u_x[\cdot]$ is done as follows: let $C$ be an array with $|\Sigma|$ elements, all initialised to $0$. As $\Sigma$ is considered to be the set of numbers $\{1,\ldots,|\Sigma|\}$, we will consider that $C$ is indexed by the letters of $\Sigma$. Also, initialise the variable $h$ with $|\Sigma|$. 

While $h>0$, we consider the positions $j$ of $x$ from the right to the left, i.e., from $n$ downwards. When reading $x[j]$, we set $C[x[j]]=j$, and if $C[x[j]]$ was $0$ before setting it to $j$, then we decrement $h$ by $1$. As soon as we have $h=0$ we stop. At this point we have $C[\ta]=g_\ta[j]$ for all $\ta\in \Sigma$, so $C[\ta]$ is the leftmost occurrence of $\ta$ to the right of $j$, and $x[j..n]$ is the shortest suffix of $x$ that contains all letters of $\Sigma$. We can set $u_x[j']=+\infty$, for all $j'>j$, and $u_x[j]=\max\{C[\ta]\mid \ta\in \Sigma\}$. 

Now let $m=u_x[j]$, and $d=j+1$ ($x[d..n]$ is the longest suffix of $x$ which is not universal). 

For $i$ from $j-1$ downto $1$ we do the following. If $m\neq C[x[i]]$ (i.e., $x[i]$ is not the same as the letter whose leftmost occurrence in $x[i+1..n]$ was the rightmost among all letters of $\Sigma$), we just set $C[x[i]]=i$. If $m=C[x[i]]$ (i.e., $x[i]$ is the same as the letter whose leftmost occurrence in $x[i+1..n]$ was the rightmost among all letters of $\Sigma$), we first set $C[x[i]]=i$ and then we need to recompute $m$, the maximum of $C$ (the position of the letter whose first occurrence in $x[i..n]$ is the rightmost among all letters). To do this, we decrement $m$ by~$1$ repeatedly, until it reaches a value $p$ such that $C[x[p]]=p$. At that point, $m=p$ is the leftmost position on which the letter $x[m]$ occurs in $x[i..n]$, and all letters of $\Sigma$ occur in $x[i..m]$. In this way, we ensure that $C[\ta]=g_\ta[i]$ for all $\ta \in \Sigma$ and $m$ points to the maximum element of $C$. In both cases, we set $u_x[i]=m$, and repeat the process for $i-1$. 

At the end of the computation described above, we computed $u_x[j]$ for every position $j$ of $x$, i.e., we know for each position $j$ of $x$ the shortest universal prefix of $x[j..n]$. The computation described above runs in time $O(n)$. For each value $j$ we set $C[x[j]]$ in constant time and then, if needed, recompute the value of $m$; this last part is not carried in constant time for each $j$, but in total $m$ traverses only once the entire word $x$ from right to left, so, summing the time spent to update $m$ for all values of $j$, we still get $O(n)$ time in total.

We now move on to the main phase of our algorithm. For $j\in [n]$, we want to compute $t_x[j]=\max\{t\mid \ScatFact_t(x[j..n])= \Sigma^t\}$ and $m_x[j]=\min\{i\mid \ScatFact_{t_x[j]}(x[j..i])= \Sigma^{t_x[j]}\}$.  

We show how to compute $m_x[j]$ and $t_x[j]$ for all positions $j$ of $x$, in $O(n)$ total time, by a simple dynamic programming algorithm. For $j\geq d$, we have $t_x[j]=0$ and $m_x[j]=u_x[j]$. For smaller values of $j$, we have $m_x[j]=u_x[j]+m_x[{u_x[j]+1}]$ and $t_x[j]=1+t_x[{u_x[j]+1}]$. Indeed, the maximum exponent $t_x[j]$ such that $\Sigma^{t_x[j]} =\ScatFact_{t_x[j]}(x[j..n])$ is obtained by taking the shortest prefix $x[j..u_x[j]]$ of $x[j..n]$ that contains all letters of $\Sigma$, and returning $1$ plus the maximum exponent $t_x[{u_x[j]+1}]$ such that $\Sigma^{t_x[{u_x[j]+1}]}$ is included in the set of scattered factors of the suffix $x[u_x[j]+1..n]$ that follows $x[j..u_x[j]]$. The value $m_x[j]$ is computed according to a similar idea. Clearly, computing $m_x[j]$ and $t_x[j]$ takes constant time for each $j$, so linear time overall. \qed
\end{proof}

%% file: theomin.tex
\begin{proof}
Consider the word $x=ww$. In a preprocessing phase, using Lemma \ref{help1}, we compute in $O(|x|)=O(n)$ time the values $t_x[j]$ and $m_x[j]$ for $j\in [2n]$.

We want to compute the minimum $\ell$ such that $w^\ell$ is $k$-universal. The general idea is the following: for $p\geq 1$, we compute the largest value $i_p$ such that $\Sigma^{i_p}=\ScatFact_{i_p}(w^p)$ as well as the shortest prefix $w^{p-1}w[1..s_p]$ of $w^p$ which is $i_p$-universal (as each $w$ contains all letters of $\Sigma$, it is clear that the shortest prefix of $w^{p}$ which is $i_p$-universal must extend inside the \nth{$p$} $w$). These values can be computed for a certain $p$ using the corresponding values for $p-1$ and the arrays we constructed in the preprocessing phase: $i_p=i_{p-1}+t_x[{s_{p-1}}]$ and $s_p=s_{p-1}+m_x[{s_{p-1}}]-n$. Essentially, for each $p$, we just extend to the right in $w^{p}$, as much as we can, the shortest prefix with the desired property constructed for $w^{p-1}$. In a simple version of our algorithm we could do that until $i_p\geq k$ (which happens after at most $k$ iterations), and return $p$ as the value we are searching for. However, this would lead to an algorithm with running time $O(n+\ell\log k/\log n)$ (where the $\log k//\log n$ factor comes from the fact that the operands in each addition $i_p=i_{p-1}+t_x[{s_{p-1}}]$ may have up to $\log k$ digits). As $\ell \leq k$ and it is natural to assume that $k$ is given in its binary representation, this algorithm could be exponential in the worst case.

We can optimise the idea above to work faster by exploiting the periodicity that occurs in the sequence $(s_p)_{p\in \N}$, defined for the repetitions of word $w$. By the pigeonhole principle, there always exist $p_1,p_2\leq n+1$ such that $s_{p_1}=s_{p_2}$. So, while $p\leq n+1$ we compute $i_p$ and $s_p$, as above, but keep track of the values taken by  $s_p$ and stop this loop as soon as the current $s_p$ has the same value as some previously computed $s_{p'}$ or $i_p\geq k$ (in the latter case, we proceed as above, and return $p$ as the value $\ell$ we look for). More precisely, we use an array $S$ with $n$ elements, all set initially to $0$. After computing $s_p$, if $S[s_p]=0$ then we set $S[s_p]=p$; if $S[s_p]\neq 0$ we proceed as follows. We stop the loop and compute two values $p_1=S[s_p]$ and $p_2=p$. It is immediate that $p_2$ is the smallest $p$ such $s_{p_1}=s_{p}$ and there are no other $p,p'<p_2$ such that $s_p=s_{p'}.$ Computing $p_1$ and $p_2$ takes $O(n)$ time. Note that all arithmetic operations we did so far are done on numbers that fit in constant memory.

Assume now that we have computed $p_2=p_1+\delta$ and $i_{p_2}=i_{p_1}+d$. It is clear that, for all $j \geq 0$, we have $s_{p_1+j\delta}=s_{p_1}$ and $i_{p_1+j\delta}=i_{p_1} + jd$. Now, let $m=k-i_{p_1}$ and $g=\lfloor\frac{m}{d}\rfloor$. Computing these numbers takes $O(\log k/\log n)$ time. 

Let $p_3=p_1 + g\delta$ (again, we need $O(\log k/\log n)$ time to compute $p_3$). We have $s_{p_3}=s_{p_1}$ and $i_{p_3}=i_{p_1} + gd \leq k$ (these operations take $O(\log k/\log n)$ time). Also, $i_{p_3+d}> k$. Let $z=k-i_{p_3}$ (and we have $z \leq d$). So, for $p$ from $p_3$ to $p_3 + \delta$ we proceed as follows. If $i_p-i_{p_3}\geq z$ (i.e., $i_p\geq k$), return $p$ as the value $\ell$ we search for. Otherwise, compute $i_{p+1}-i_{p_3}=(i_{p}-i_{p_3})+t_{s_{p}}$ (in time $O(1)$ as it can be done with only adding numbers which are smaller than $d$) and $s_{p+1}=s_{p}+m_{s_{p}}-n$, and iterate. Because we certainly reach, in this loop, a $p$ such that $i_p\geq k$, and $\delta\leq n$, the execution of the loop takes $O(n)$ time. 

Hence, we get the smallest $\ell$ such that $w^\ell$ is $k$-universal (i.e., $i_\ell\geq k$), in $O(n+\log k/\log n)$ time. \qed
\end{proof}

%% file: appendixlastpart.tex
For $\ell,n \in \N$ and $w_1,\dots,w_n\in\Sigma^{\ast}$, 
define $\lsk 
w_1,\dots,w_n\rsk_\ell$ as the set of all words $w=x_1\dots x_{\ell}$ with 
$x_i\in\{w_1,\dots,w_n\}$ and $\lsk 
w_1,\dots,w_n\rsk=\bigcup_{\ell \in\N}\lsk 
w_1,\dots,w_n\rsk_\ell$.

\begin{definition}
Let $n\in\N$.
The set $S=\{w_1,\dots,w_n|\,w_i\in\Sigma^*,i\in[n]\}$ is 
{\em $k$-universal} if there exists $u\in\lsk w_1,\dots,w_n\rsk$ such that $u$ is $k$-universal.
\end{definition}

Firstly we need introduce some notation for convenience and to prove an auxiliary lemma.
To each $S\subseteq \Sigma$ we associate a word $u_S$ with $|u_S|=|S|$ and $\letters(u_S)=S$ (i.e., $u_S$ is a linear ordering of the letters from $S$). Following the notations from Lemma \ref{help1}, for a word $x$, let $t_x=\max\{t\in\N_0\mid \ScatFact_t(x)= \Sigma^t\}$ and $m_x=\min\{i\in \N_0\mid \ScatFact_{t_x}(x[1..i])= \Sigma^{t_x}\}$; clearly, if $t_x=0$, then $m_x=0$, too. Note now that, for a word $u$ with $\letters(u)=S$ and $|u|=|S|$, we have $t_{u_Sw}=t_{uw}$ and $m_{u_Sw}=m_{uw}$, for all $w\in \Sigma^*$. Consider $w_1,\ldots,w_p\in\Sigma^{\ast}$, and take $j \in [p]$. For $\ell\in \N$ and $S'\subset \Sigma$, we define $\max_\ell(S,j,S')=\max\{t_{w}\mid w=u_Sw'w_j, w'\in \lsk w_1,\ldots,w_p\rsk_{\ell-1}$ and $\letters(w[m_{w}+1..|w|])= S' \}.$ By the remarks regarding the choice of the word $u_S$,  $\max_\ell(S,j, S')$ is clearly well defined. 

\begin{lemma}\label{greedy}
For $w_1,\ldots,w_p\in\Sigma^{\ast}$, $S\subseteq \Sigma$, $\ell\in\N_{\geq 2}$, and all $\ell'\in[\ell-1]$, we have $\max_\ell (S,j, S')= \max\{\max_{\ell'} (S,k,S'') + \max_{\ell-\ell'} (S'',j,S'))\mid k\in [p], S''\subseteq \Sigma \}$.
\end{lemma}

\begin{proof}
Let $\ell'$ be a natural number such that $1\leq \ell'<\ell$. 
Let ${i_1},\ldots,{i_\ell}\in [p]$ such that ${i_\ell}=j$, $\max_\ell(S,j,S')=t_{w}$, and $\letters(w[m_{w}+1.. |w|])=S'$, for $w=u_Sw_{i_1}\cdots w_{i_\ell}$. Take $x'=u_Sw_{i_1}\cdots w_{i_{\ell'}}$, $S''=\letters({x'[m_{x'}..|x'|]})$, and $x''=u_{S''}w_{i_{\ell'+1}}\cdots w_{i_\ell}$. It is not hard to see that $\max_\ell(S,j,S')=t_{x'}+ t_{x''}$. 

Assume that $\max_{\ell'}(S,i_{\ell'}, S'') > t_{x'}$. Let $h_1,\ldots,h_{\ell'}\in [p]$ and $w_{h_1},\ldots,w_{h_{\ell'}}\in \Sigma^*$ be such that $h_{\ell'}=i_{\ell'}$, $\max_{\ell'}(S,h_{\ell'},S'')=t_{v'}$, and $\letters(v'[m_{v'}+1..|v'|])=S'$, for $v'=u_Sw_{h_1}\cdots w_{h_{\ell'}}$. Then, for $v''=u_Sw_{h_1}\cdots w_{h_{\ell'}}  w_{i_{\ell'+1}}\cdots w_{i_\ell}$ we have $t_{v''}>t_w=\max_\ell(S,j,S')$, a contradiction. Thus, $\max_{\ell'}(S,i_{\ell'}, S'') = t_{x'}$. We can similarly show that $\max_{\ell-\ell'}(S'',j, S') = t_{x''}$.

 Assume now that there exists $r\in [p]$ and $T\subseteq \Sigma$ such that $\max_{\ell'}(S,r,T) + \max_{\ell-\ell'}(T,j,S')>t_{x'}+ t_{x''}=t_w$. Let $j_1,\ldots,j_\ell\in [p]$ and $w_{j_1},\ldots,w_{j_{\ell}}\in \Sigma^*$ be such that $j_{\ell'}=r$, $j_\ell=j$, $\max_{\ell'}(S,j_{\ell'},T)=t_{x}$ and $\letters(x[m_x+1..|x|])=T$, for $x=u_Sw_{j_1}\cdots w_{j_{\ell'}}$, and $\max_\ell(T,j_{\ell'},S')=t_{y}$ and $\letters(y[m_y+1..|y|])=S'$, for $y=u_Tw_{j_{\ell'+1}}\cdots x_{j_{\ell}}$. Then, clearly, for $v=u_Sx_{j_1}\cdots x_{j_{\ell}}$ we have $t_v>t_w=\max_{\ell}(S,j,S')$, a contradiction, considering the form of $v$. 
\end{proof}

\begin{theorem}\label{basis2}
Given $w_1,\ldots,w_p\in\Sigma^{\ast}$ with $|w_1\cdots w_p|$ $=n$
and $\letters(w_1\cdots
w_p)=\Sigma$, and $k\in\N$, 
we can compute the minimal $\ell$ 
for which there exist $\{i_1,\ldots,i_\ell\} \subseteq [k]$ such that 
$w_{i_1}\cdots w_{i_\ell}$ is $k$-universal in $O(2^{3|\Sigma|} p^2 \log \ell + n)$ steps, some being arithmetic operations on numbers with $\log k$ bits. The overall time complexity of our algorithm is $O(\frac{\log k}{\log n}(2^{3|\Sigma|} p^2 \log \ell + n))$.
\end{theorem}

\begin{proof}
Note first that, because $\Sigma=\letters(w_1\cdots w_p)$, we have $\ell\leq pk$. We can now sketch the algorithm computing $\ell$. The general idea is first to construct, by dynamic programming, concatenations of $2^e$ factors of the set $\{w_1,\ldots,w_p\}$, for larger and larger $e$, until we find one such concatenation with $2^f$ elements that is $k'$-universal, for some $k'\geq k$. That is, we compute the values $N_e[S,S',j]=\max_{2^e}(S,i,S')$, for $e$ from $0$ until we reach an array $N_f$ which contains a value $N_f[\emptyset,S',j]\geq k$. Note that $2^f$ is the smallest power of $2$ such that the concatenation of $2^f$ numbers is $k$-universal, so, consequently, $f\leq 2\ell$. The values in each of the array $N_e$ are computed by dynamic programming based on the values in the array $N_{e-1}$, using Lemma \ref{greedy} for $\ell=2^e$ and $\ell'=2^{e-1}$. Once this computation is completed, we use binary search to obtain the exact value of $\ell$, as required. However, we now have the benefit that we can perform this binary search in an interval upper bounded by $2^f\in O(\ell)$. 

In the following we describe the algorithm in details. We will evaluate its complexity first as the number of steps (including arithmetic operations on numbers with up to $\log k$ bits) it performs. Then we compute its actual time complexity.

We start with a preprocessing phase. We construct the $p \times |\Sigma|$ matrix $F[\cdot,\cdot]$, indexed by the numbers between $1$ and $p$ and the letters of $\Sigma$ (which are numbers between $1$ and $|\Sigma|$). We have $F[i,a]$ is the position of the first (i.e., leftmost) occurrence of each letter $x\in \Sigma$ in $w_i$. This matrix can be computed as follows. Initialise all elements of $F$ with $0$. For each $i$, we traverse $w_i$, letter by letter, from left to right. When the letter $x\in \Sigma$ is read on position $j$ of $w_i$, if $F[i,z]=0$ then we set $F[i,x]=j$. The total number of steps needed to do this is $O(|\Sigma|p + n)$ (as it includes the initialisation of $F$). Similarly, we construct the $p \times |\Sigma|$ matrix $L[\cdot,\cdot]$, indexed by the numbers between $1$ and $p$ and the letters of $\Sigma$, where $L[i,a]$ is the position of the rightmost occurrence of each letter $x\in \Sigma$ in $w_i$. 

We compute also in the preprocessing phase the data structures from Lemma \ref{help1}, for each word $w_i$, with $i\in [p]$. So, we have for each word $w_i$ the arrays $t_{w_i}[j]=\max\{t\mid \ScatFact_t(w_i[j..n])= \Sigma^t\}$ and $m_{w_i}[j]=\min\{g\mid \ScatFact_{t_{w_i}[j]}(w_i[j..g])$ $= \Sigma^{t_{w_i}[j]}\}$. This is done in $O(n)$ steps.

Then, for each set $S\subseteq \Sigma$ and $i\in [p]$, we compute in $O(\Sigma)$, the value $j=\max\{F[x,i]\mid x\in \Sigma\setminus S\}$. Basically, $w_i[1..j]$ is the shortest prefix of $w_i$ such that $u_Sw_i$ contains all letters of $\Sigma$. Let $g=m_{w_i}[j+1]$, and let $S'\subseteq \Sigma$ be the set of letters contained by $w_i[g+1..|w_i|]$. The set $S'$ can be computed in $O(\Sigma)$ time, by selecting in $S'$ the letters $x\in \Sigma$ with $L[i,x]>g$. We set $M_{1}[S,i]=(1+t_{w_i}[j], S')$, where $M_1$ is an additional matrix we use. The computation of $M_1[S,i]$, performed for a set $S$ and a number $i\in [p]$, takes $O(|\Sigma|)$. So, in total we compute the matrix $M_1$ in $O(2^{|\Sigma|}|\Sigma|p)$ time. It is worth noting that if $M_1[S,i]=(h,S')$, then $\max_{1}(S,i,S')=h$.

The main phase of the algorithm follows. If there is an element $M_1[\emptyset,i]=(h,S')$ such that $h\geq k$, then we return $\ell=1$. If not we proceed as described next.

For $e$ natural number such that $\log_{pk}+1 \geq e \geq 1$, we define the $3$-dimensional array $N_e[\cdot,\cdot,\cdot]$, whose first two indices are subsets of $\Sigma$, and the third is a number from $[p]$, and $N_e[S,S',i]=\max_{2^e}(S,i,S')$. That is, $N_e[S,S',i]$ stores the maximum $k$ such that there exists $k$-universal word $w$ which is the concatenation of $u_S$ followed by $2^e$ words from $\{w_1,\ldots,w_p\}$, ending with $w_i$, and, moreover, if $w'$ is the suffix of $w$ that follows the shortest $k$-universal prefix of $w$, then $\letters(w')=S'$. The elements $N_e[S,S',i]$ will be computed by dynamic programming, using Lemma \ref{greedy} for $\ell=2^e$ and $\ell'=\frac{\ell}{2}$.

For $e=1$, the elements of the array $N_e$ are computed as follows. By Lemma \ref{greedy}, $N_1[S,S',i]=\max\{g\mid g=g_1+g_2$ where $M_1[S,j]=(g_1,S'')$ and $M_1[S'',i]=(g_2,S'), $ with $j\in [p], S''\subseteq \Sigma \}$.
For $e>1$, we have $N_e[S,S',i]=\max\{g\mid g=g_1+g_2$ where $N_{e-1}[S,S'',j]=g_1$ and $N_{e-1}[S'',S',i]=g_2, $ with $j\in [p], S''\subseteq \Sigma \}$. Clearly, computing each of the arrays $N_e$ takes $O(2^{3|\Sigma|}p^2)$. It is not hard to see that the maximum element of $N_e$ is strictly greater than the maximum element of $N_{e-1}$. 

We stop the computation of the arrays $N_e$ as soon as we reach such an array $N_f$ such that there exists $i$ and $S'$ with $N[\emptyset, S',i]\geq k$. We get that $2^{f-1}< \ell \leq 2^f$ (where $\ell$ is the value we want to compute), so the total time needed to compute all these arrays is $O(2^{3|\Sigma|}p^2 \log \ell)$. 

Now we need to search $\ell$ between $b=2^{f-1}$ and $s=2^{f}$.  We will do this by an adapted binary search. Denote $N'=N_{f-1}$ and $N''=N_f$. Let $h$ be maximum such that $b+2^h< s$. We compute the $3$-dimensional array $N_{mid}[\cdot,\cdot,\cdot]$, indexed just as the arrays $N_e$. We have $N_{mid}[S,S',i]=\max\{g\mid g=g_1+g_2$ where $N'[S,S'',j]=g_1$ and $N_{h}[S'',S',i]=g_2, $ with $j\in [p], S''\subseteq \Sigma \}$. If $N_{mid}$ contains an element greater or equal to $\ell$, we repeat this search for the same $b$ and $N'$, and $s=b+2^h$ and $N''=N_{mid}$. Otherwise, we repeat the search for the same $s$ and $N''$, and using $b+2^h$ instead of $b$ and $N_{mid}$ instead of $N'$. We stop the process if $b=s-1$, and return $s$. This procedure is iterated $O(f)=O(\log \ell)$ times. Thus, computing $\ell$ is done in $O(2^{3|\Sigma|}p^2 \log \ell)$ steps, some of which are arithmetic operations on numbers with up to $\log k$ bits.

The overall complexity of the algorithm is, thus, $O(\frac{\log k}{log n}(2^{3|\Sigma|}p^2 \log \ell) + n)$. \qed
\end{proof}

Note that, in the case stated in the previous theorem, computing the minimal 
number of words (from a given set) that should be concatenated in order to 
obtain a $k$-universal word is fixed parameter tractable w.r.t. the parameter 
$|\Sigma|$, the size of the alphabet of the input words. If both $p$, the number 
of input words, and $|\Sigma|$ are constant, the algorithm runs in $O(n+\log 
\ell)$ steps, which is linear w.r.t. the size of the input because $\log \ell 
\leq \log (pk)=\log p + \log k$ (but the overall time is still affected by the operations on large numbers). In fact, we can give a solution with a linear number of steps 
for this problem in the case of words over binary alphabets (and $p$ is 
not necessarily constant). The main idea is, in this case, we can show that, from an input 
set of words, only a constant number are useful when trying to construct a 
$k$-universal word by a minimal number of concatenations. The following result is based on the arch factorisation and Proposition~\ref{decomp}.

\begin{theorem}\label{basis-binary}
Given $k\in\N$ and $w_1,\ldots,w_p\in\{\ta,\tb\}^{\ast}$ with $\letters(w_1\cdots w_p)=\{\ta,\tb\}$ and $|w_1\cdots w_p|=n$, 
we can compute in $O(n+\log \ell)$ steps the minimal $\ell$ for which there 
exist $\{i_1,\ldots,i_\ell\} \subseteq [k]$ such that $w_{i_1}\cdots w_{i_\ell}$ 
is $k$-universal. The overall complexity of the algorithm is, thus, $O(\frac{\log k}{log n}\log \ell + n)$. 
\end{theorem}

\begin{proof}
Let $u_0\in \{w_1,\ldots,w_p\}$ be such that $t_{u_0} \geq t_{w_i}$, for all $i\in [p]$. For each $x \in \{\ta,\tb\}$, let $u_x\in \{w_1,\ldots,w_p\}$ be such that $u_x$ starts with $x$ and $t_{u_x[2..|u_x|]} \geq t_{w_i[2..|w_i|]}$, for all $i\in [p]$. 
For each $x \in \{\ta,\tb\}$, let $v_x\in \{w_1,\ldots,w_p\}$ be such that $v_x$ ends with $x$ and $t_{v_x[1..|v_x|-1]} \geq t_{w_i[1..|w_i|-1]}$, for all $i\in [p]$. 
For each pair $x,y \in \{\ta,\tb\}$, let $u_{x,y}\in \{w_1,\ldots,w_p\}$ be such that $u_{x,y}$ starts with $x$ and ends with $y$ and $t_{v_x[2..|v_x|-1]} \geq t_{w_i[2..|w_i|-1]}$, for all $i\in [p]$. 
In case of equalities, we just any word that fulfils the desired property.

Let $S=\{u_0\}\cup \{u_x\mid x\in \{\ta,\tb\}\} \cup \{v_x\mid x\in \{\ta,\tb\}\} \cup \{u_{x,y}\mid x,y\in \{\ta,\tb\}\}$. Clearly, $S$ contains at most $9$ words. Note that all words of $S$ can be computed in $O(n)$ time, using the same strategy as in Proposition~\ref{decomp}. 

One can show that for every concatenation of $m$ words from $\{w_1,\ldots,w_p\}$ which is $k$-universal, there exist a concatenation of $m$ words from $S$ which is $k'$-universal, for some $k'\geq k$. Thus, it is enough to solve the problem for the input set $S$, of constant size, instead of the whole $\{w_1,\ldots,w_p\}$. Hence, by Theorem \ref{basis2}, the conclusion follows. 

Indeed, let $w=w_{i_1}\cdots w_{i_{\ell-1}} w_{i_\ell} w_{i_{\ell+1}} \cdots w_{i_{m}} $, such that $w_{i_\ell}\not\in S$. To compute $t_w=t$ we can proceed as in Proposition~\ref{decomp} and identify $t$ factors $d_1,\ldots,d_t \in \{\ta\tb,\tb\ta\}$ of $w$ such that $w=s_0d_1s_1\cdots d_ts_t$, where $s_i\in \{\ta,\tb\}^*$ for $i\in [t]_0$. Clearly, $|\letters(s_i)|\leq 1$, for all $i \in [t]_0$. Now, we do a case analysis. 

Let $x=w_{i_\ell}[1]$ and $y=w_{i_\ell}[|w_{i_\ell}|]$. If the first letter of $w_{i_\ell}$ is the last letter of a factor $d_i$ and the last letter of $w_{i_\ell}$  is the first letter of a factor $d_j$ (with $i<j$), let $w'=w_{i_1}\cdots w_{i_{\ell-1}} u_{x,y} w_{i_{\ell+1}} \cdots w_{i_{m}} $; it is immediate that $t_{w'}\geq t_{w}$. If the first letter of $w_{i_\ell}$ is the last letter of some $d_i$ but the last letter of $w_{i_\ell}$ is not the first letter of any factor $d_j$ (where $j>i$), let $w'=w_{i_1}\cdots w_{i_{\ell-1}} u_{x} w_{i_{\ell+1}} \cdots w_{i_{m}} $; it is immediate that $t_{w'}\geq t_{w}$. If the first letter of $w_{i_\ell}$ is not the last letter of any factor $d_i$ but the last letter of $w_{i_\ell}$ is the first letter of a factor $d_j$, let $w'=w_{i_1}\cdots w_{i_{\ell-1}} u_{y} w_{i_{\ell+1}} \cdots w_{i_{m}} $; it is immediate that $t_{w'}\geq t_{w}$. Finally, if the first letter of $w_{i_\ell}$ is not the last letter of any factor $d_i$ and the last letter of $w_{i_\ell}$  is not the first letter of any factor $d_j$, let $w'=w_{i_1}\cdots w_{i_{\ell-1}} u_{0} w_{i_{\ell+1}} \cdots w_{i_{m}} $; it is immediate that $t_{w'}\geq t_{w}$.

So, if a concatenation of $m$ words $w_{i_1}\cdots  w_{i_{m}}$ is $t$-universal, we could iteratively replace all the words which are not part of $S$ by words of $S$ and obtain a concatenation with $m$ input words, which is $t'$-universal, with $t'\geq t$. Therefore, to solve the problem from the statement of the theorem is enough to produce the set $S$ and then solve the problem for the input set $S$ instead of $\{w_1,\ldots,w_p\}$. For that we can use the algorithm from Theorem \ref{basis2}, which will run in $O(n+\frac{\log k \log \ell}{\log n} )$ steps, because both $S$ and $\Sigma$ are constant.\qed 
\end{proof} 

In a particular case of Theorem \ref{basis2} each of the input words contain all letters of $\Sigma$. Once again, we obtain a polynomial algorithm.

\begin{theorem}\label{basis}
Given $w_1,\ldots,w_p\in\Sigma^{\ast}$, with $\letters(w_i)=\Sigma$ for all $i\in [p]$ and $|w_1\cdots w_p|=n$, and $k\in\N$, 
we can compute in polynomial time $O(n+p^3|\Sigma|\log\ell \frac{\log k}{log n})$ the minimal $\ell$ for which there exist $\{i_1,\ldots,i_\ell\} \subseteq [k]$ with $w_{i_1}\cdots$ $w_{i_\ell}$ is $k$-universal.
\end{theorem}

The proofs of Theorems \ref{basis2} and \ref{basis} are based on 
a common dynamic programming algorithm: for all subsets $S\subset \Sigma$ and 
increasing values of an integer $\ell>1$, we compute the maximal universality 
index $m$ that we can obtain by concatenating $2^t$ words from the input set such 
that the respective concatenation consists in a prefix which is $m$-universal, followed by a suffix over $S$. 
we stop as soon we reach an $m \geq k$ as 
universality index. We then optimise the number of concatenated words needed to obtain universality index $k$ by binary search. 
Now, for Theorem \ref{basis2} we really have to consider all the sets $S$, in each step, while in the case of Theorem 
\ref{basis} it is enough to consider only the sets that occur as alphabets of 
the suffixes of the input words. This is why this strategy can be implemented 
more efficiently in the case when all input words are universal to begin with.

\begin{proof} (of Theorem~\ref{basis})
We follow the idea of the algorithm of Theorem \ref{basis2}: construct, by dynamic programming, longer and longer concatenations of factors of the set $\{w_1,\ldots,w_p\}$, until one such concatenation which is $k$-universal is obtained. The main difference is that in each concatenation $w= w_{i_1}\ldots w_{i_m}$, the shortest prefix of $w$ which is $k$-universal ends inside $w_{i_m}$, because $\letters(w_i)=\Sigma$ for all $i\in [p]$. As such, the $\ell$ we search for is at most $k$, but also this allows us to get rid of the exponential dependency on $\Sigma$ from Theorem \ref{basis2}, as we can now work with certain suffixes of the words $w_{i}$, instead of subsets of $\Sigma$, when defining our dynamic programming structures.
Informally, our algorithm works as follows: we find the highest universality index of a concatenation of $2^e$ words of $\{w_1,\ldots,w_p\}$, which starts inside $w_i$ and ends inside $w_j$, for all $i$ and $j$, and suitable starting and, respectively, ending positions. This can be efficiently computed for several reasons. Firstly, such a concatenation is obtained by putting together two concatenations of roughly $2^{e-1}$ words of $\{w_1,\ldots,w_p\}$ which have the highest universality index, the first starting in the same place within $w_i$, followed by $2^{e-1}-2$ words of the input set, and ending with a prefix of length $c$ of some $w_q$, and the second one starting with $w_q[c+1..|w_q|]$ followed by $2^{e-1}$ words from the input set, ending in the same place within $w_j$. Secondly, a concatenation of $2^e$ words of $\{w_1,\ldots,w_p\}$ with the highest universality~index, ending inside $w_j$, can actually only end on some very specific positions of $w_j$: the positions where each letter of $\Sigma$ occurs for the first time in the shortest prefix of $w_j$ that contains all letters of $\Sigma$. Consequently, the starting positions of such concatenations (useful in our algorithm either directly as solutions, or as building blocks for larger concatenations) are also restricted. Putting these two ideas together, and using an adapted binary search to search for $\ell$, we obtain an algorithm with the stated complexity. 

Once again, we start with some preliminaries and a preprocessing phase. Let $\sigma=|\Sigma|$. 

To begin with, let us consider a concatenation $w= w_{i_1}\cdots w_{i_m}$, and let $t$ be the maximum number such that $w$ is $t$-universal. By Lemma \ref{decomp} we can greedily decompose $w=d_1\cdots d_td'$, such that $\letters(d_j)=\Sigma$, $\letters(d')$ is a strict subset of $\Sigma$, and $d_1\cdots d_j$ is the shortest prefix of $w$ which is $j$-universal, for all $j\in [t]$. Because $\letters(w_i)=\Sigma$ for all $i$, we have that each factor $d_j$ is either fully contained in one of the words $w_{i_g}$, for $j\in [t]$ and $g\in [m]$, or it starts inside $w_{i_{g}}$ and ends inside $w_{i_{g+1}}$, for some $g\in [m]$. In the following, we call a factor $d_j$ crossing if it starts  inside $w_{i_{g}}$ and ends inside $w_{i_{g+1}}$, for some $g\in [m]$. If $d_j$ is such a factor, then $d_j$ can only start on some very specific positions inside $w_{i_g}$. Firstly, the suffix of $w_{i_g}$ that comes after $d_{j-1}$ cannot contain all letters of $\Sigma$. So $d_{j-1}$ must end inside the shortest suffix of $w_{i_g}$ that contains all letters of $\Sigma$. Assume this suffix starts on position $r$ and note that it starts with the last occurrence of some letter of $\Sigma$ in $w_{i_g}$. So, $d_{j-1}$ ends on a position $r'\geq r$. Due to the greedy construction of $d_{j-1}$, it follows that the last letter of $d_{j-1}$ occurs in $w[r..r']$ exactly once. So, $d_{j-1}$ ends on the first occurrence of a letter of $\Sigma$ to the right of $r$. There are at most $\sigma$ such positions. Consequently, $d_j$ starts exactly on the next position after $d_{j-1}$ ends, and we also have at most $\sigma$ positions where $d_j$ may start.

In conclusion, in each word $w_i$, part of a concatenation  $w= w_{i_1}\ldots w_{i_m}$, there are at most $\Sigma$ positions where a crossing factor can start. Each crossing factor $d_j$ is constructed by appending to $d_j$ (in a left to right traversal, from the starting position of the factor) the letters of the considered concatenation until $\Sigma=\letters(d_j)$. Therefore, $d_j$ is uniquely determined by the two factors it crosses ($w_{i_g}$ and $w_{i_{g+1}}$) and its starting position inside $w_{i_g}$. Hence, there can be at most $O(p^2|\Sigma|)$ crossing factors overall, and we will determine all of them in our preprocessing.

In the preprocessing phase, we first construct the $p \times |\Sigma|$ matrix $F[\cdot,\cdot]$, indexed by the numbers between $1$ and $p$ and the letters of $\Sigma$ (which are numbers between $1$ and $|\Sigma|$). We have $F[i,a]$ is the position of the first (i.e., leftmost) occurrence of each letter $x\in \Sigma$ in $w_i$. This matrix can be computed as follows. Initialise all elements of $F$ with $0$. For each $i$, we traverse $w_i$, letter by letter, from left to right. When the letter $x\in \Sigma$ is read on position $j$ of $w_i$, if $F[i,z]=0$ the we set $F[i,x]=j$. The total number of steps needed to do this is $O(|\Sigma|p + n)$ (as it includes the initialisation of $F$). Similarly, we construct the $p \times |\Sigma|$ matrix $L[\cdot,\cdot]$, indexed by the numbers between $1$ and $p$ and the letters of $\Sigma$, where $L[i,a]$ is the position of the rightmost occurrence of each letter $x\in \Sigma$ in $w_i$. Using $L[i,\cdot]$ we also determine the position $r_i$ of $w_i$ such that $w_i[r_i..|w_i|]$ is the shortest suffix of $w_i$ that contains all letters of $\Sigma$. Also, in another traversal of $w_i$ we can compute the increasingly sorted list $L_i$ of positions where each letter of $\Sigma$ occurs for the first time in $w_i[r_i..|w_i|]$. More precisely, we construct the lists $L_i=(i_1,x_1),\ldots,(i_\sigma,x_\sigma)$, where $i_g<i_{g+1}$ for $g\in \Sigma$, and $\{x_1,\ldots,x_\sigma\}=\Sigma$. The needed to compute all these structures is $O(n)$. 

Now, we compute the factors crossing from $w_i$ to $w_j$. They should start on one of the positions $i_1+1$, $i_2+1, \ldots$ $i_{\sigma}+1$, obtained using $L_i$. Let $c_{i,j}[i_g+1]$ be the crossing factor that starts on position $i_g+1$ for some $g\in [\sigma]$. The prefix of $c_{i,j}[i_g+1]$ contained in $w_i$ contains only the letters $x_{g+1},\ldots,x_\sigma$ and none of the letters $x_1,\ldots,x_g$. Thus, $c_{i,j}[i_g+1]$ extends in $w_j$ until it contains all the missing letters, i.e., till the maximum position among $F[j,x_1]$, $F[j,x_2], \ldots$, $F[j,x_g]$. This observation allows us to compute the respective crossing factors efficiently. Let $C[i,j,g]$ be the last position (inside $w_j$) of $c_{i,j}[i_{g}+1]$. Then $C[i,j,1]=F[j,x_1]$. For $g>1$, $C[i,j,g]=\max\{F[j,x_g],C[i,j,g-1]\}$. 

The time needed to compute the values $C[i,j,\cdot]$ is $O(\sigma)$. We do this computation for all $i$ and $j$, and, as such, we identify the starting and ending positions for all possible crossing factors in $O(p^2\sigma)$.

Still in the preprocessing phase, we compute the data structures from Lemma \ref{help1}, for each word $w_i$, with $i\in [p]$. So, we have for each word $w_i$ the arrays $t_{w_i}[j]=\max\{t\mid \ScatFact_t(w_i[j..n])= \Sigma^t\}$ and $m_{w_i}[j]=\min\{g\mid \ScatFact_{t_{w_i}[j]}(w_i[j..g])= \Sigma^{t_{w_i}[j]}\}$. Let $t_{w_i}=t_{w_i}[1]$ and $m_{w_i}=m_{w_i}[1]$. This takes $O(n)$ time.

Further, we present the main phase of our algorithm, that computes the value $\ell$ for which there exist $\{i_1,\ldots,i_\ell\} \subseteq [k]$ such that $w_{i_1}\cdots w_{i_\ell}$ is $k$-universal.

Firstly, if there exists $i$ such that $t_{w_i}\geq k$, we have $\ell=1$. Otherwise, we continue as follows. 

For $e\in [k]$, $e\geq 1$, we define the $3$-dimensional arrays $R_e[\cdot,\cdot,\cdot]$, whose first and third indices are numbers $i,j\in [p]$, and second index is a number from $\{0\}\cup L_i$ (so each $R_e$ has size $O(p^2\sigma)$. We define $R_e[i,j,c]=(t,d)$ where $t$ is the maximum number for which there exist ${i_2},\ldots,{i_{2^e-1}}\in [p]$ such that $t_w=t$, where $w=w_{i}[c+1..|w_i|]w_{i_2}\cdots w_{i_{2^e-1}}w_j$, and $d$ is the minimum number for which there exist ${i_2},\ldots,{i_{2^e-1}}$ such that $t_w=t$, where $w=w_{i}[c+1..|w_i|]w_{i_2}\cdots w_{i_{e-1}}w_j[1..d]$. That is, $R_e[i,j,c]$ stores, on its first component, the maximum $t$ such that there exists $t$-universal word $w$ which is the concatenation of the suffix of $w_i$ that starts on position $c+1$, followed by $2^e-2$ words from the set $\{w_1,\ldots,w_p\}$, and then followed by $w_j$. Moreover, $R_e[i,j,c]$ also stores, on its second component, the minimum value $m_w$ obtained for a concatenation $w=w_{i}[c+1..|w_i|]w_{i_2}\cdots w_{i_{2^e-1}}w_j$ for which $t_w=t$ (i.e., $t_w$ is as large as possible). We define also the $3$-dimensional arrays $P_e[\cdot,\cdot,\cdot]$, exactly as the above with the single difference that in the definition of the elements of $P_e$ we consider the concatenation of $2^e+1$ elements, not just $2^e$ as it was the case for $R_e$.

The elements $R_e[i,j,c]$ and $P_e[i,j,c]$ can be computed by dynamic programming, somehow similarly to the approach of Theorem \ref{basis2}. To simplify the exposure we also define the array $R_0[\cdot,\cdot,\cdot]$, in which only the elements $R_0[i,i,c-1]=(t_{w_i}[c],m_{w_i}[c])$ are defined (the others are set to $-\infty$). Clearly, $R_0$ can be computed in $O(p^2\sigma)$. 

To describe the general computation, we need to compare pairs of numbers. We say that $(a,b)$ is {\em more useful} than $(c,d)$ if $a>b$ or $a=b$ and $c\leq d$. Also, if $p=(a,b)$ is a pair of natural numbers, then its first projection is $\pi_1(p)=a$ and its second projection is $\pi_2(p)=b$.

To compute $R_1$ we can use the formula:
\[R_1[i,j,c-1] = (t_{w_i}[c] + 1 + \pi_1(R_0[j,j, C_{i,j}[m_{w_i}[c]+1]]), m_{w_j}[1+C_{i,j}[m_{w_i}[c]+1]]),\] 
for $i,j\in [p]$ and $c\in \{0\}\cup L_i$. 
\begin{figure}\begin{center}
\includegraphics[width=0.8\linewidth]{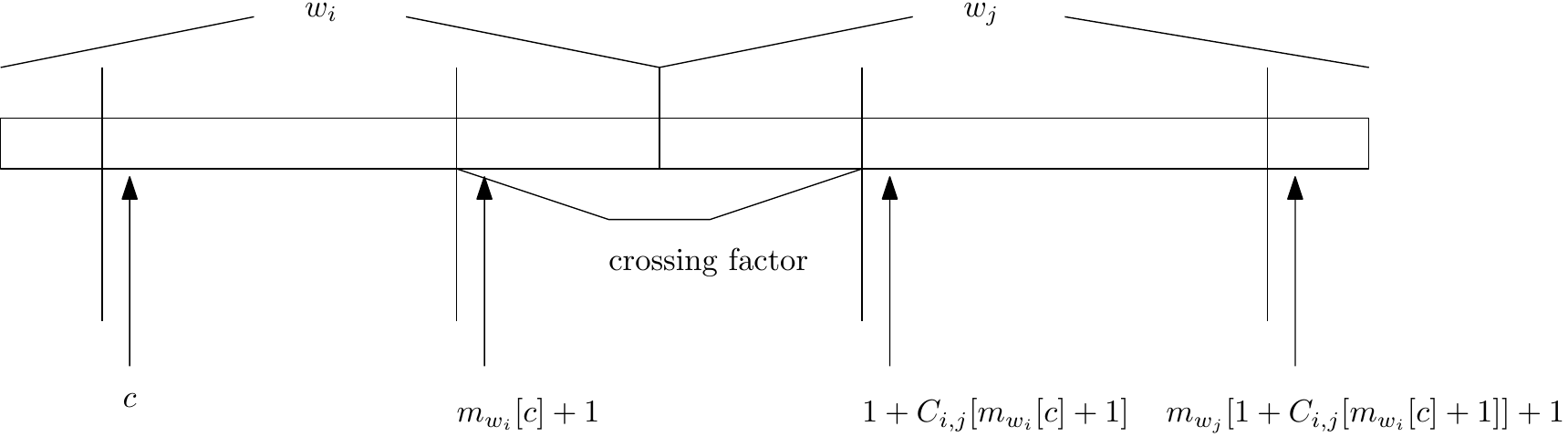}
\caption{The computation of $R_1[i,j,c-1]$}
\end{center}
\end{figure}

Indeed, when computing $R[i,j,c-1]$ we start on position $c$ of $w_i$ and essentially try to identify as many consecutive strings whose alphabet is $\Sigma$ in the concatenation of $w_i$ and $w_j$ as possible. Firstly. using $t_{w_i}[c] $ and $m_{w_i}[c]$ we find the shortest factor $w_i[c..m_{w_i}[c]]$ which has the highest universality index among all factors of $w_i$ starting on position $c$. Then we use the crossing factor that corresponds to $m_{w_i}[c]$ to move into $w_j$, on position $c'=C_{i,j}[m_{w_i}[c]+1]$, and then find the shortest factor $w_j[c'..m_{w_j}[c']]$ which has the highest universality index among all factors of $w_j$ starting on position $c'$. Following similar arguments to those in the proof of Lemma \ref{decomp} we get that $R_0$ is correctly computed in this way: our strategy here corresponds exactly to the greedy strategy employed in the respective lemma. 

After we compute $R_e$, for some $e\geq 1$, we first compute $P_e$. The formula for the elements of $P_e$ is given in the following. Let $q\in [p]$ be such that $R_e[q,j, 1+m_{w_q}[1+C_{i,q}[m_{w_i}[c]+1]]]$ is more useful than any other pair $R_e[q',j, 1+m_{w_{q'}}[1+C_{i,q'}[m_{w_i}[c]+1]]]$ for $q'\in [p]$. We then can compute \\
 $P_e[i,j,c-1] = (t_{w_i}[c] + 1 + t_{w_q}[1+C_{i,q}[m_{w_i}[c]+1]] + \pi_1(R_e[q,j, m_{w_q}[1+C_{i,q}[m_{w_i}[c]+1]]]),$ \\
\hspace*{2.4cm}$ \pi_2(R_e[q,j, m_{w_q}[1+C_{i,q}[m_{w_i}[c]+1]]]),$\\
 for $i,j\in [p]$ and $c\in \{0\}\cup L_i$. 
 \begin{figure}\begin{center}
\includegraphics[width=0.9\linewidth]{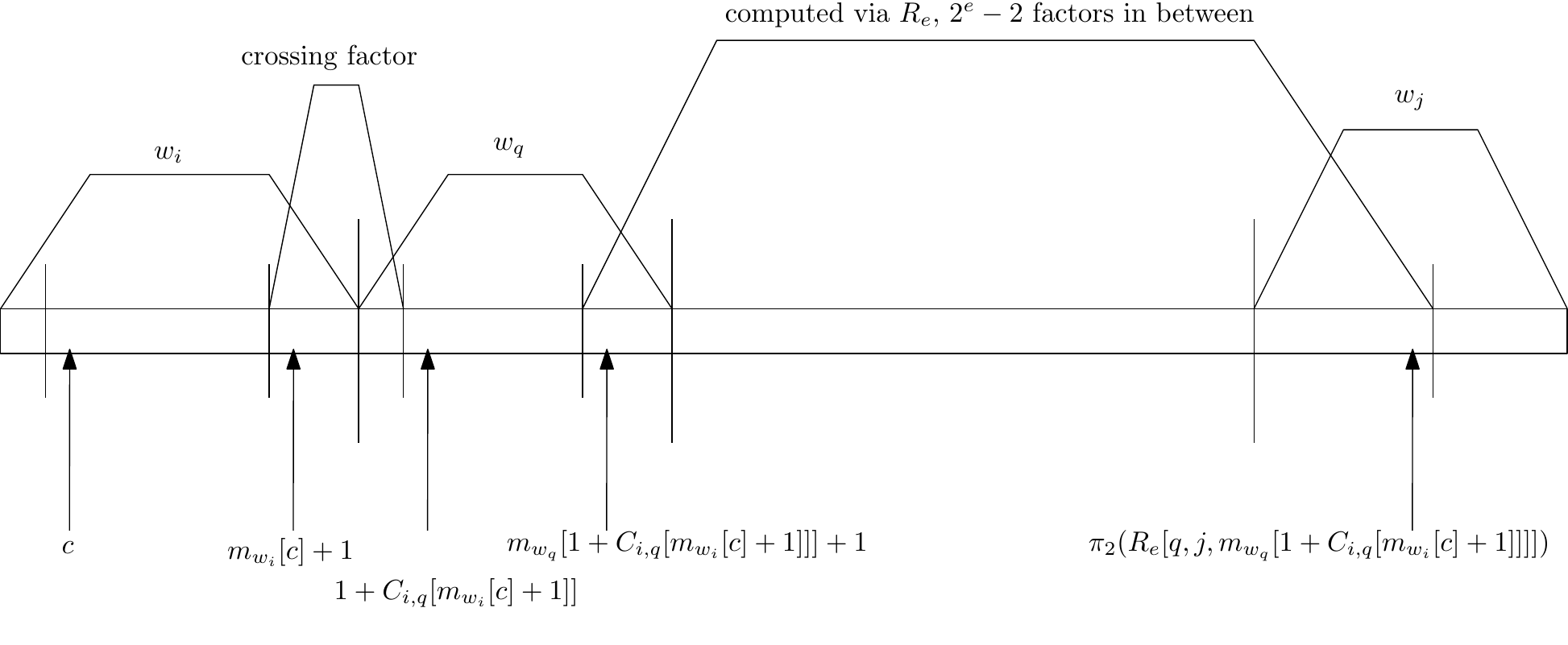}
\caption{The computation of $P_e[i,j,c-1]$}
\end{center}
\end{figure}
 Similarly to the computation of $R_1$, when computing $P_e[i,j,c-1]$ we start on position $c$ of $w_i$ and try to add to $w_i[c..|w_i|]$ a concatenation of $2^e$ words of $\{w_1,\ldots,w_p\}$ (ending with $w_j$), which contains as many consecutive strings, whose alphabet is $\Sigma$, as possible. This is done using the greedy approach of Lemma \ref{decomp}.  As such, we use $t_{w_i}[c] $ and $m_{w_i}[c]$ we find the shortest factor $w_i[c..m_{w_i}[c]]$ which has the highest universality index among all factors of $w_i$ starting on position $c$. Then we identify the word $w_q$, such that after using the crossing factor that corresponds to $m_{w_i}[c]$ to move from $w_i$ into $w_q$ we can reach $w_j$ by concatenating another $2^e-2$ factors in between, to obtain a word with the highest universality index among all such possible concatenations. Once again, it is not hard to see that this formula is correct (see also the figure below). Firstly, the choice of the factor $w_i[c..m_{w_i}[c]]$ as the first group of consecutive strings, each with the alphabet $\Sigma$, is correct due to the greedy approach in Lemma~\ref{decomp}. Then, we need to cross into the rest of the factors in the concatenation of words from $\{w_1,\ldots,w_p\}$. For each choice $w_{q'}$ of the factor following $w_i$ in this concatenation, we cross into this word from $w_i$ in an optimal way: we use the crossing string ending on $C_{i,q'}[m_{w_i}[c]+1]$. Any shorter word would not work, any longer word does not make sense due to the greedy strategy of Lemma \ref{decomp}. Then, using the already computed $m_{w_q}$ and $R_e$ we start from $1+C_{i,q'}[m_{w_i}[c]+1]$ and follow the optimal selection of the concatenated strings given by these arrays. We then select from all these possibilities (computed for each $q'$) the one that produces a string with higher universality index. So, the computation of $P_e[i,j,c-1] $ is correct. 

After computing $P_e$ for some $e\geq 1$, we compute $R_{e+1}$. For some $i,j\in [p]$ and $c$ with $c\in \{0\}\cup L_i$, let $q\in [p]$ be such that $\pi_1(R_e[i,q,c])+\pi_1(P_e[q,j,\pi_2(R_e[i,q,c])]) \geq \pi_1(R_e[i,q',c])+\pi_1(P_e[q',j,\pi_2(R_e[i,q',c])])$ for all $q'\in [p]$. To break equalities, we select $q$ such that $\pi_2(P_e[q,j,\pi_2(R_e[i,q,c])])$ is minimal. Then, we can compute $R_{e+1}[i,j,c]$ by the following formula:\\
$R_{e+1}[i,j,c-1]= (\pi_1(R_e[i,q,c-1])+\pi_1(P_e[q,j,\pi_2(R_e[i,q,c-1])]),$\\
\hspace*{2.9cm}$ \pi_2(P_e[q,j,\pi_2(R_e[i,q,c-1])])),$  for $i,j\in [p]$ and $c\in \{0\}\cup L_i$.  \\
The idea is pretty similar to how we computed the other arrays. 
 \begin{figure}\begin{center}
\includegraphics[width=0.9\linewidth]{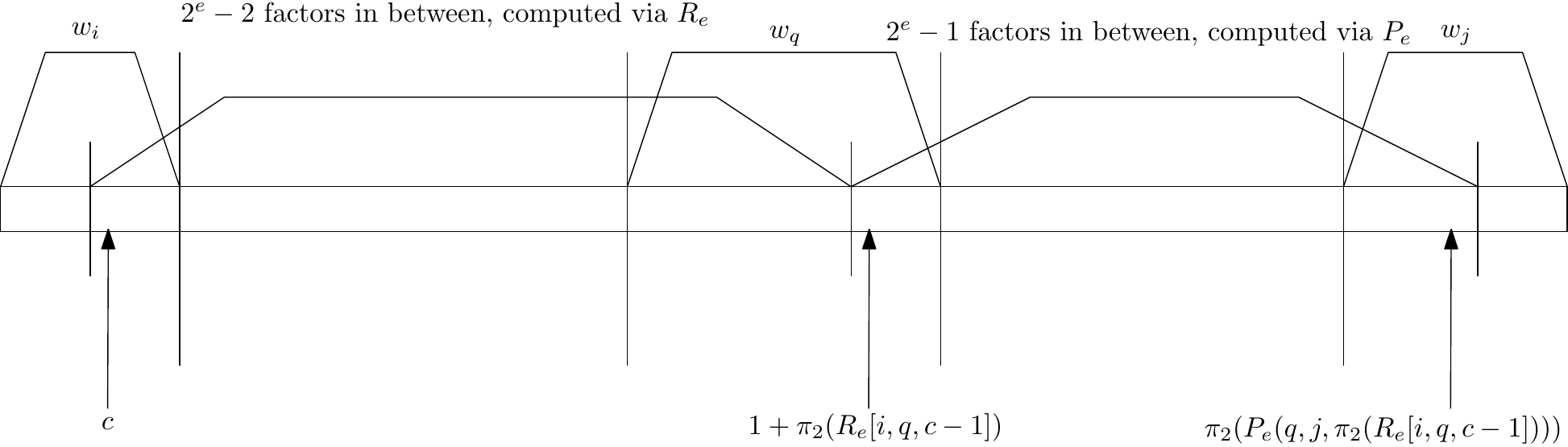}
\caption{The computation of $R_{e+1}[i,j,c-1]$}
\end{center}
\end{figure}
We start on position $c+1$ of $w_i$ and try to add to $w_i[c+1..|w_i|]$ a concatenation of $2^{e+1}-1$ words of $\{w_1,\ldots,w_p\}$ (ending with $w_j$), which contains as many consecutive strings, whose alphabet is $\Sigma$, as possible. We iterate over all possible choices for the $2^e$-th word in this concatenation, namely $w_{q}$. We use the value computed found in $R_e(i,q,c)$ to find the concatenation of $2^e$ words with highest universality index that starts with $w_i[c..|w_i|]$ and ends with $w_{q'}$. Then we continue this concatenation again in the best way (i.e., by the concatenation of $2^{e+1}$ words with the highest universality index), as given by $P_e[q',j,\pi_2(R_e[i,q',c])]$. Then we just take the value $q$ for which we obtained the most useful pair $(\pi_1(R_e[i,q,c])+\pi_1(P_e[q,j,\pi_2(R_e[i,q,c])]), \pi_2(P_e[q,j,\pi_2(R_e[i,q,c])]) $. Once more, the greedy approach shown to be correct in Lemma \ref{decomp} proves that the formula used for the elements of $R_{e+1}$ is also correct. 

Clearly, the complexity of computing each element of $P_e$ and $R_e$ is $O(p)$. So, computing each of these matrices takes $O(p^3\sigma)$. 

As in the proof of Theorem \ref{basis2}, we stop as soon as we computed an array $R_f$ that contains an element $R_f[i,j,0]$ with $\pi_1(R_f[i,j,0])\geq k$. We have $f\in O(\log \ell)$. 

 Now we need to search $\ell$ between $b=2^{f-1}$ and $s=2^{f}$. And we can proceed exactly as in the proof of the aforementioned theorem, by an adapted binary search. Denote $R'=R_{f-1}$ and $R''=R_f$. Let $h$ be maximum such that $b+2^h< s$. We compute the $3$-dimensional array $R_mid[\cdot,\cdot,\cdot]$, indexed just as the arrays $R_e$. We have \\
$R_{mid}[i,j,c-1]= (\pi_1(R'[i,q,c-1])+\pi_1(P_h[q,j,\pi_2(R'[i,q,c-1])]), $\\
\hspace*{2.9cm}$\pi_2(P_h[q,j,\pi_2(R'[i,q,c-1])])),$ 
for $i,j\in [p]$ and $c\in \{0\}\cup L_i$. 

If $R_{mid}$ contains an element whose first component is greater or equal to $\ell$, we repeat this search for the same $b$ and $R'$, and $s=b+2^h$ and $R''=R_{mid}$. Otherwise, we repeat the search for the same $s$ and $R''$, and using $b+2^h$ instead of $b$ and $R_{mid}$ instead of $R'$. We stop the process if $b=s-1$, and return $s$. This procedure is iterated $O(f)=O(\log \ell)$ times. 

The overall number of steps of the algorithm we described is, thus, $O(p^3 \sigma \log \ell + n)$. Of course, in the part where we compute concatenations with large universality index we need to manage arithmetic operations with $\log k$-bit numbers. So, our algorithm requires $O(p^3 \sigma \log \ell \frac{\log k}{\log n} + n)$ time. \qed
\end{proof}

%% file: del-pref-suf.tex
\begin{proof}
To compute the longest prefix $w[1..i-1]$ of $w$ which has universality index $\ell$, we use data structures from Lemma \ref{help1}. We start with $j=1$ and $k=0$. While $k\neq \ell+1$ do $t=u_w[j]$, increase $k$, set $j=t+1$. If $k=\ell+1$ then $w[1..t]$ is the shortest prefix of $w$ which is $\ell+1$ universal. Therefore  the longest prefix $w[1..i-1]$ of $w$ which has universality index $\ell$ is $w[1..t-1]$. A similar approach can be used for suffixes. \qed
\end{proof}

%% file: conclusion.tex
In this paper we investigated the scattered factor universality of words. We have proven how this universality behaves if a word is repeated and how this characterisation can be exploited to obtain linear-time algorithms for obtaining an uncommon scattered factor. Moreover we set the universality of a palindrome into relation with its first half (minus one letter if the length is odd) as well as the generalised repetition $w\pi(w)$ for a morphic permutation $\pi$. The last part of Section~\ref{basic} dealt with circular universality. Here we have proven the relation between universality and circular universality and we have proven that the characterisation in Theorem~\ref{theoitercirc} does not hold for arbitary alphabets. We conjecture that for an alphabet of cardinality $\sigma$ the notion of circularity has to be generalised such that, assuming the word as a circle, not once but $\sigma-1$ times the word has to read before the universality is increased. Finally in the last section we developed data structures that allow us to determine the universality of factors of a given word.

%% file: scatteredUniversalityDLT.bbl
\begin{thebibliography}{10}

\bibitem{BringmannK18}
Karl Bringmann and Marvin K{\"{u}}nnemann.
\newblock Multivariate fine-grained complexity of longest common subsequence.
\newblock In {\em Proc. {SODA} 2018}, pages 1216--1235. {SIAM}, 2018.

\bibitem{ChenKMS17}
Herman Z.~Q. Chen, Sergey Kitaev, Torsten M{\"u}tze, and Brian~Y. Sun.
\newblock On universal partial words.
\newblock {\em Electronic Notes in Discrete Mathematics}, 61:231--237, 2017.

\bibitem{cormen}
Thomas~H. Cormen, Charles~E. Leiserson, Ronald~L. Rivest, and Clifford Stein.
\newblock {\em Introduction to Algorithms, 3rd Edition}.
\newblock {MIT} Press, 2009.

\bibitem{crochemore}
Maxime Crochemore, Christophe Hancart, and Thierry Lecroq.
\newblock {\em Algorithms on strings}.
\newblock Cambridge University Press, 2007.

\bibitem{dlt2019}
Joel~D. Day, Pamela Fleischmann, Florin Manea, and Dirk Nowotka.
\newblock k-spectra of weakly-c-balanced words.
\newblock In {\em Proc. {DLT} 2019}, volume 11647 of {\em Lecture Notes in
  Computer Science}, pages 265--277. Springer, 2019.

\bibitem{Bruijn46}
Nicolaas~G. de~Bruijn.
\newblock A combinatorial problem.
\newblock {\em Koninklijke Nederlandse Akademie v. Wetenschappen}, 49:758--764,
  1946.

\bibitem{LucaGZ08}
Aldo de~Luca, Amy Glen, and Luca~Q. Zamboni.
\newblock Rich, sturmian, and trapezoidal words.
\newblock {\em Theor. Comput. Sci.}, 407(1-3):569--573, 2008.

\bibitem{dobkin}
David~P. Dobkin and Richard~J. Lipton.
\newblock On the complexity of computations under varying sets of primitives.
\newblock {\em J. Comput. Syst. Sci.}, 18(1):86--91, 1979.

\bibitem{DroubayJP01}
Xavier Droubay, Jacques Justin, and Giuseppe Pirillo.
\newblock Episturmian words and some constructions of de {L}uca and {R}auzy.
\newblock {\em Theor. Comput. Sci.}, 255(1-2):539--553, 2001.

\bibitem{ElzingaRW08}
Cees~H. Elzinga, Sven Rahmann, and Hui Wang.
\newblock Algorithms for subsequence combinatorics.
\newblock {\em Theor. Comput. Sci.}, 409(3):394--404, 2008.

\bibitem{KufMFCS}
Lukas Fleischer and Manfred Kufleitner.
\newblock Testing {S}imon's congruence.
\newblock In {\em Proc. {MFCS} 2018}, volume 117 of {\em LIPIcs}, pages
  62:1--62:13. Schloss Dagstuhl - Leibniz-Zentrum fuer Informatik, 2018.

\bibitem{FreydenbergerGK15}
Dominik~D. Freydenberger, Pawel Gawrychowski, Juhani Karhum{\"{a}}ki, Florin
  Manea, and Wojciech Rytter.
\newblock Testing k-binomial equivalence.
\newblock {\em CoRR}, abs/1509.00622, 2015.

\bibitem{gabow}
Harold~N. Gabow and Robert~Endre Tarjan.
\newblock A linear-time algorithm for a special case of disjoint set union.
\newblock In {\em Proc. 15th STOC}, pages 246--251, 1983.

\bibitem{GawrychowskiRSS17}
Pawel Gawrychowski, Narad Rampersad, Jeffrey Shallit, and Marek Szykula.
\newblock Existential length universality.
\newblock {\em to appear at STACS}, abs/1702.03961, 2020.

\bibitem{GoecknerGHKKKS18}
Bennet Goeckner, Corbin Groothuis, Cyrus Hettle, Brian Kell, Pamela
  Kirkpatrick, Rachel Kirsch, and Ryan~W. Solava.
\newblock Universal partial words over non-binary alphabets.
\newblock {\em Theor. Comput. Sci}, 713:56--65, 2018.

\bibitem{HalfonSZ17}
Simon Halfon, Philippe Schnoebelen, and Georg Zetzsche.
\newblock Decidability, complexity, and expressiveness of first-order logic
  over the subword ordering.
\newblock In {\em Proc. {LICS} 2017}, pages 1--12, 2017.

\bibitem{TCS::Hebrard1991}
Jean-Jacques Hebrard.
\newblock An algorithm for distinguishing efficiently bit-strings by their
  subsequences.
\newblock {\em Theoretical Computer Science}, 82(1):35--49, 22~May 1991.

\bibitem{HolzerK11}
Markus Holzer and Martin Kutrib.
\newblock Descriptional and computational complexity of finite automata - {A}
  survey.
\newblock {\em Inf. Comput.}, 209(3):456--470, 2011.

\bibitem{union-find}
Hiroshi Imai and Takao Asano.
\newblock Dynamic segment intersection search with applications.
\newblock In {\em Proc. 25th Annual Symposium on Foundations of Computer
  Science, FOCS}, pages 393--402. {IEEE} Computer Society, 1984.

\bibitem{KarandikarKS15}
Prateek Karandikar, Manfred Kufleitner, and Philippe Schnoebelen.
\newblock On the index of {S}imon's congruence for piecewise testability.
\newblock {\em Inf. Process. Lett.}, 115(4):515--519, 2015.

\bibitem{CSLKarandikarS}
Prateek Karandikar and Philippe Schnoebelen.
\newblock The height of piecewise-testable languages with applications in
  logical complexity.
\newblock In {\em Proc. {CSL} 2016}, volume~62 of {\em LIPIcs}, pages
  37:1--37:22, 2016.

\bibitem{journals/lmcs/KarandikarS19}
Prateek Karandikar and Philippe Schnoebelen.
\newblock The height of piecewise-testable languages and the complexity of the
  logic of subwords.
\newblock {\em Logical Methods in Computer Science}, 15(2), 2019.

\bibitem{KrotzschMT17}
Markus Kr{\"{o}}tzsch, Tom{\'{a}}s Masopust, and Micha{\"{e}}l Thomazo.
\newblock Complexity of universality and related problems for partially ordered
  {NFA}s.
\newblock {\em Inf. Comput.}, 255:177--192, 2017.

\bibitem{KuskeZ19}
Dietrich Kuske and Georg Zetzsche.
\newblock Languages ordered by the subword order.
\newblock In {\em Proc. {FOSSACS} 2019}, volume 11425 of {\em Lecture Notes in
  Computer Science}, pages 348--364. Springer, 2019.

\bibitem{Rigo19}
Marie Lejeune, Julien Leroy, and Michel Rigo.
\newblock Computing the k-binomial complexity of the {T}hue-{M}orse word.
\newblock {\em CoRR}, abs/1812.07330, 2018.

\bibitem{LeroyRS17a}
Julien Leroy, Michel Rigo, and Manon Stipulanti.
\newblock Generalized {P}ascal triangle for binomial coefficients of words.
\newblock {\em CoRR}, abs/1705.08270, 2017.

\bibitem{Loth97}
M.~Lothaire.
\newblock {\em Combinatorics on Words}.
\newblock Cambridge University Press, 1997.

\bibitem{Maier:1978}
David Maier.
\newblock The complexity of some problems on subsequences and supersequences.
\newblock {\em J. ACM}, 25(2):322--336, April 1978.

\bibitem{martin1934}
Monroe~H. Martin.
\newblock A problem in arrangements.
\newblock {\em Bull. Amer. Math. Soc.}, 40(12):859--864, 12 1934.

\bibitem{Mat04}
Alexandru Mateescu, Arto Salomaa, and Sheng Yu.
\newblock Subword histories and {P}arikh matrices.
\newblock {\em Journal of Computer and System Sciences}, 68(1):1--21, 2004.

\bibitem{Rampersad:2012}
Narad Rampersad, Jeffrey Shallit, and Zhi Xu.
\newblock The computational complexity of universality problems for prefixes,
  suffixes, factors, and subwords of regular languages.
\newblock {\em Fundam. Inf.}, 116(1-4):223--236, January 2012.

\bibitem{RigoS15}
Michel Rigo and Pavel Salimov.
\newblock Another generalization of abelian equivalence: Binomial complexity of
  infinite words.
\newblock {\em Theor. Comput. Sci.}, 601:47--57, 2015.

\bibitem{Salomaa05}
Arto Salomaa.
\newblock Connections between subwords and certain matrix mappings.
\newblock {\em Theoretical Computer Science}, 340(2):188--203, 2005.

\bibitem{Seki12}
Shinnosuke Seki.
\newblock Absoluteness of subword inequality is undecidable.
\newblock {\em Theor. Comput. Sci.}, 418:116--120, 2012.

\bibitem{Simon72}
Imre Simon.
\newblock Piecewise testable events.
\newblock In {\em Autom.\ Theor.\ Form.\ Lang., 2nd GI Conf.}, volume~33 of
  {\em LNCS}, pages 214--222. Springer, 1975.

\bibitem{Wagner:1974}
Robert~A. Wagner and Michael~J. Fischer.
\newblock The string-to-string correction problem.
\newblock {\em J. ACM}, 21(1):168--173, January 1974.

\bibitem{Zetzsche16}
Georg Zetzsche.
\newblock The complexity of downward closure comparisons.
\newblock In {\em Proc. {ICALP} 2016}, volume~55 of {\em LIPIcs}, pages
  123:1--123:14, 2016.

\end{thebibliography}
